%% file: main.tex
\documentclass[11pt]{article}
\input{macro.tex}

\begin{document}
\title{Randomized polynomial-time equivalence between determinant and trace-IMM equivalence tests}
\author{Janaky Murthy\\
\footnotesize{Indian Institute of Science}\\
\normalsize{\tt janakymurthy@iisc.ac.in}
\and {Vineet Nair}\\
\footnotesize{Technion Israel Institute of Technology}
\footnote{A part of this work was done when the author was a graduate student at the Indian Institute of Science.}\\ 
\normalsize{\tt vineet@cs.technion.ac.il}
\and {Chandan Saha}\\
\footnotesize{Indian Institute of Science}\\
\normalsize{\tt chandan@iisc.ac.in}}
\clearpage\maketitle
\thispagestyle{empty}

\begin{abstract}
Equivalence testing for a polynomial family $\{g_m\}_{m \in \N}$ over a field $\F$ is the following problem: Given black-box access to an $n$-variate polynomial $f(\vecx)$, where $n$ is the number of variables in $g_m$ for some $m \in \N$, check if there exists an $A \in \GL(n,\F)$ such that $f(\vecx) = g_m(A\vecx)$. If yes, then output such an $A$. The complexity of equivalence testing has been studied for a number of important polynomial families, including the determinant ($\Det$) and the family of iterated matrix multiplication polynomials. 
Two popular variants of the iterated matrix multiplication polynomial are: $\oIMM_{w,d}$ (the $(1,1)$ entry of the product of $d$ many $w\times w$  symbolic matrices) and $\IMM_{w,d}$ (the trace of the product of $d$ many $w\times w$ symbolic matrices). The families -- $\Det$, $\oIMM$ and $\IMM$ -- are $\mathsf{VBP}$-complete under $p$-projections, and so, in this sense, they have the same complexity. But, do they have the same equivalence testing complexity? We show that the answer is ``yes" for $\Det$ and $\IMM$ (modulo the use of randomness).


The above result may appear a bit surprising as the complexity of equivalence testing for $\oIMM$ and that for $\Det$ are quite different over $\Q$: a randomized polynomial-time  equivalence testing for $\oIMM$ over $\Q$ is known \cite{KayalNST19}, whereas \cite{GargGKS19} showed that equivalence testing for $\Det$ over $\Q$ is integer factoring hard (under randomized reductions and assuming GRH). 
To our knowledge, the complexity of equivalence testing for $\IMM$ was not known before this work. We show that, despite the syntactic similarity between $\oIMM$ and $\IMM$, equivalence testing for $\IMM$ and that for $\Det$ are randomized polynomial-time Turing reducible to each other over any field of characteristic zero or sufficiently large. 
The result is obtained by connecting the two problems via another well-studied problem in computer algebra, namely the \emph{full matrix algebra isomorphism} problem ($\FMAI$). In particular, we prove the following: 
\begin{enumerate}
    \item Testing equivalence of polynomials to $\IMM_{w,d}$, for $d\geq 3$ and $w\geq 2$, is randomized polynomial-time Turing reducible to testing equivalence of polynomials to $\Det_w$, the determinant of the $w \times w$ matrix of formal variables. (Here, $d$ need not be a constant.)
    \item $\FMAI$ is randomized polynomial-time Turing reducible to equivalence testing (in fact, to tensor isomorphism testing) for the family of \emph{matrix multiplication tensors} $\{\IMM_{w,3}\}_{w \in \N}$.
\end{enumerate}
These results, in conjunction with the randomized poly-time reduction (shown in \cite{GargGKS19}) from determinant equivalence testing to $\FMAI$, imply that the four problems -- $\FMAI$, equivalence testing for $\IMM$ and for $\Det$, and the $3$-tensor isomorphism problem for the family of matrix multiplication tensors  -- are randomized poly-time equivalent under Turing reductions.

\end{abstract}

\newpage
\pagenumbering{arabic}
\input{sec_intro.tex}
\input{sec_prelim.tex}
\input{sec_Lie_algebra_trace.tex}
\input{sec_reduction_to_det_eq.tex}

\input{sec_reduction_from_fmai_to_trace_equiv.tex}

\section*{Acknowledgments}
We are thankful to Avi Wigderson for his suggestion on designing an equivalence testing algorithm for $\IMM$ at the end of VN's presentation at CCC 2017. We would also like to thank Christian Ikenmeyer for his question on equivalence testing for $\IMM$ which encouraged us to work on this problem. Thanks also to Neeraj Kayal and Ankit Garg for helpful discussions, and particularly to Neeraj for pointing us to \cite{GrowchowQ19}.
VN is thankful to be funded by the European Union’s Horizon 2020 research and innovation programme under grant agreement No 682203 -ERC-[ Inf-Speed-Tradeoff].

\bibliographystyle{alpha}
\bibliography{references}

\appendix
\input{secappendix_ABPprelim}

\input{secappendix_lie_algebra_IMM.tex}

\input{secappendix_reduction_from_trace_to_mult_equivalence.tex}

\input{secappendix_reduction_to_det_eq.tex}
\input{sec_reduction_from_mult_trace_to_deg3_mult_trace.tex}

\input{secappendix_reduction_from_FMAI_to_IMM.tex}


\end{document}

%% file: macro.tex
\usepackage{fullpage}
\usepackage{amssymb}
\usepackage{amsmath}
\usepackage{amsthm}
\usepackage{multirow}
\usepackage{enumerate}
\usepackage{graphicx}
\usepackage{mathrsfs}
\usepackage[utf8]{inputenc}
\usepackage{cite}
\usepackage{hyperref}
\usepackage{cleveref}
\usepackage{mathtools}
\usepackage{amsfonts}
\usepackage[T1]{fontenc}
\usepackage{mathpazo}
\usepackage{bm}
\usepackage{isomath}
\usepackage{verbatim}
\usepackage{changepage}
\usepackage[usenames,dvipsnames,svgnames,table]{xcolor}
\usepackage [english]{babel}
\usepackage{url}

\definecolor{blueviolet}{RGB}{60,50,200}
\definecolor{oliveg}{RGB}{40,200,30}
\hypersetup{colorlinks=true,       
    linkcolor=blueviolet,
    citecolor=oliveg,
}

\usepackage{tikz}
\usetikzlibrary{calc}

\usepackage{algorithmicx,algorithm}
\usepackage{algpseudocode}

\usepackage{color}

\usepackage{todonotes}

\theoremstyle{definition}
\newtheorem{definition}{Definition}[section]

\theoremstyle{plain}
\newtheorem{lemma}{Lemma}[section]
\newtheorem{theorem}{Theorem}
\newtheorem{fact}{Fact}

\newtheorem{corollary}{Corollary}[section]
\newtheorem{claim}{Claim}[section]

\newtheorem{observation}{Observation}[section]

\newcommand{\alert}[1]{\textcolor{red}{#1}}

\newcommand{\diag}{\textnormal{diag}}
\newcommand{\tr}{\textnormal{tr}}

\newcommand{\CAL}[1]{\mathcal{#1}}
\newcommand{\G}{\mathfrak{g}}

\newcommand{\F}{\mathbb{F}}

\newcommand{\N}{\mathbb{N}}
\newcommand{\Z}{\mathbb{Z}}
\newcommand{\Q}{\mathbb{Q}}
\newcommand{\C}{\mathbb{C}}

\newcommand{\poly}{\textnormal{poly}}
\renewcommand{\char}{\textnormal{char}}
\newcommand{\Det}{\text{Det}}
\newcommand{\veca}{{\mathbf{a}}}
\newcommand{\vecb}{{\mathbf{b}}}
\newcommand{\vecc}{{\mathbf{c}}}

\newcommand{\vecx}{{\mathbf{x}}}
\newcommand{\vecy}{{\mathbf{y}}}
\newcommand{\vecz}{{\mathbf{z}}}
\newcommand{\vecu}{{\mathbf{u}}}

\newcommand{\vecr}{{\mathbf{r}}}

\newcommand{\spa}{\mathrm{span}}

\newcommand{\GIMM}{\mathfrak{g}_{\scriptscriptstyle \IMM}}

\newcommand{\oIMM}{\text{IMM}}
\newcommand{\IMM}{\text{Tr}\mbox{-}\text{IMM}}
\newcommand{\EIMM}{\mathsf{TRACE}}
\newcommand{\EDET}{\mathsf{DET}}
\newcommand{\FMAI}{\mathsf{FMAI}}
\newcommand{\MMTI}{\mathsf{MMTI}}
\newcommand{\IMMTI}{\mathsf{TRACE\mbox{-}TI}}

\newcommand{\vineet}[1]{\textcolor{red}{#1}}

\makeatletter
\newcommand\footnoteref[1]{\protected@xdef\@thefnmark{\ref{#1}}\@footnotemark}
\makeatother

\newcommand{\GL}{\text{GL}}
\newcommand{\GLNF}{\text{GL}(n,\F)}

\newcommand{\subimm}{\textnormal{\footnotesize{Tr-{\tiny IMM}}}}
\newcommand{\gimm}{\textnormal{G}_{\subimm}}
\newcommand{\Gf}{\mathfrak{g}_{f}}
\newcommand{\realimm}{\textnormal{IMM}_{w,d}}

%% file: sec_intro.tex
\section{Introduction}\label{sec: intro}
The \emph{polynomial equivalence problem} or \emph{equivalence testing} is the following algorithmic task: Given two $n$-variate polynomials $f$ and $g$ over a field $\F$ as lists of coefficients, determine if there exists an $A \in \GLNF$ such that $f(\vecx) = g(A\vecx)$. If yes, then $f$ is said to be \emph{equivalent to}\footnote{Indeed, $f$ and $g$ represent the same function on $\F^n$ upto a change of basis.} $g$ over $\F$. The complexity of equivalence testing depends on the underlying field $\F$. Over finite fields, the problem is in $\mathsf{NP} \cap \mathsf{coAM}$ \cite{Thierauf98,Saxenaphd}\footnote{This is shown by using the classic set lower bound protocol \cite{GoldwasserS86}.}, and hence unlikely to be $\mathsf{NP}$-complete. Whereas over $\Q$, it is not even known whether equivalence testing is decidable. The best known complexity of the problem over other fields follows from a naive reduction to solving a system of polynomial equations. However, polynomial solvability could be harder than testing polynomial equivalence. \\


\textbf{Connections to other problems.} A few works in the literature have related equivalence testing to other fundamental problems. For example, \cite{AgrawalS05} showed that the special instance of cubic form equivalence is at least as hard as (but possibly harder than) graph isomorphism, irrespective of the underlying field. There is a close connection between cubic form equivalence and the algebra isomorphism problem. \cite{AgrawalS06} gave a polynomial-time reduction from commutative algebra isomorphism to cubic form equivalence over any field. In the reverse direction, a polynomial-time reduction is known from cubic form equivalence to commutative algebra isomorphism over \emph{almost} all fields \cite{GrowchowQ19, AgrawalS05}. In fact, the results in \cite{BW15}, \cite{FGS19} and \cite{GrowchowQ19} together imply that a host of problems, which includes $3$-tensor isomorphism, matrix space isometry, matrix space conjugacy, (commutative or associative) algebra isomorphism and cubic form equivalence, are polynomial-time reducible to each other. There is a cryptographic authentication scheme \cite{Patarin96} based on the presumed hardness of cubic form equivalence\footnote{more generally, constant-degree form equivalence} over finite fields (or rather a generalization of it known as \emph{Isomorphism of Polynomials with one Secret} (IP1S)\footnote{IP1S is the following problem: Given two ordered sets of $n$-variate polynomials $(f_1, f_2, \ldots, f_m)$ and $(g_1, g_2, \ldots, g_m)$, decide if there exists an $A \in \GLNF$ such that $f_i(\vecx) = g_i(A\vecx)$ for all $i \in [m]$. Note that even the quadratic case is non-trivial here as we are dealing with tuples of polynomials. Recently, \cite{IvanyosQ19} gave a randomized poly-time algorithm for the quadratic IP1S problem over finite fields of odd size. In the general setting, there is an algorithm for IP1S over finite fields that is significantly better than the brute-force strategy, but it still runs in exponential time \cite{FaugereP06, PatarinGC98}.}). It is not known whether cubic form equivalence is even decidable over $\Q$. In contrast, the complexity of quadratic form equivalence testing is completely resolved, primarily due to well-known classification results for quadratic forms (see \cite{Serre73, Araujo11}). The classification yields a polynomial-time quadratic form equivalence testing over finite fields. Over $\Q$ though, quadratic form equivalence can be solved in polynomial time only with oracle access to integer factoring. Moreover, integer factoring reduces in randomized polynomial time to quadratic form equivalence over $\Q$ \cite{Walter13}\footnote{This reduction is to the search version of the quadratic form equivalence problem. In the search version of equivalence testing, we are required to output an invertible transformation $A$ if the input polynomials are equivalent.}. \\


\textbf{Special polynomial families.} The work of \cite{Kayal11} initiated the study of a natural variant of the polynomial equivalence problem, namely equivalence testing for special families of polynomials. In this setting, we fix some important family of polynomials $\CAL{G} = \{g_m\}_{m \in N}$ and then aim to design an equivalence testing algorithm for $\CAL{G}$. Such an algorithm takes input black-box access\footnote{i.e., query access to evaluations of $f$ at chosen points from $\F^n$.} to a single $n$-variate polynomial $f(\vecx)$ and determines whether $f$ is equivalent to $g_m$ for some $m \in \N$, and if yes, then it also outputs an $A \in \GL(n,\F)$ such that $f(\vecx) = g_m(A\vecx)$.\footnote{The problem is well-posed even if $f$ is given verbosely as a list of coefficients and it is not required to output an invertible transformation $A$ in the `yes' case. However, it turns out that for a number of popular polynomial families it is indeed possible to design efficient equivalence testing algorithms that satisfy these stronger requirements.} \cite{Kayal12,Kayal11} gave randomized polynomial-time equivalence testing algorithms for a few interesting polynomial families, viz. the determinant, the permanent, the family of elementary symmetric polynomials and the family of power symmetric polynomials. These families are quite popular in algebraic complexity theory, particularly in the context of proving arithmetic circuit lower bounds (see the surveys \cite{ShpilkaY10, ChenKW11, S15survey}). Except for the determinant, the algorithms in \cite{Kayal12,Kayal11} work over $\C, \Q$, and finite fields\footnote{Over $\C$, the computation model assumes that arithmetic with numbers in $\C$ and root finding of univariate polynomials over $\C$ can be done efficiently. Also, the finite fields are assumed to be of sufficiently large characteristic.}, and for the determinant it works only over $\C$. Recently, \cite{GargGKS19} gave a randomized polynomial-time equivalence testing algorithm for the determinant over finite fields\footnote{A determinant equivalence test over finite fields was also given in \cite{KayalNS19}, but the algorithm there outputs an invertible transformation over a low extension of the base field.}. They also showed that determinant equivalence test over $\Q$ is intimately connected to integer factoring:  Let $\Det_w(\vecx)$ be the determinant of the $w \times w$ symbolic matrix. Then, deciding if a given polynomial is equivalent to $\Det_w$ over $\Q$ can be done in randomized polynomial-time with oracle access to integer factoring, provided $w$ is a constant\footnote{When $w$ is not a constant, \cite{GargGKS19} gave a randomized polynomial-time determinant equivalence test over $\Q$, but the algorithm (which works without an integer factoring oracle) outputs a transformation over a low extension of $\Q$. }. Furthermore, assuming GRH, there is a randomized polynomial-time reduction from factoring square-free integers to finding an $A \in \GL(2,\Q)$ such that a given quadratic form $f = \Det_2(A \cdot \vecx)$, if $f$ is equivalent to $\Det_2$. \\ 

Determinant equivalence test is particularly interesting in the context of the permanent versus determinant problem \cite{Valiant79a}. An approach to solve this long-standing open problem is given by Geometric Complexity Theory (GCT) \cite{MulmuleyS01, MulmuleyS08}, which proposes the applications of deep tools and techniques from algebraic geometry, group theory and representation theory to achieve this goal. GCT reduces the problem to showing that the (padded) permanent polynomial is not in the \emph{orbit closure}\footnote{The orbit of an $n$-variate degree-$d$ polynomial $g \in \C[\vecx]$ is the set $\{g(A\vecx) \mid A\in \GL(n, \C) \}$, and the orbit closure of $g$ is the Zariski closure of the orbit when viewed as points in $\C^{{n+d \choose d}}$.} of a polynomial-size determinant polynomial, and suggests (among other things) to develop an algorithmic approach to do the same. Equivalence testing for the determinant is the related problem of checking if a given polynomial is in the orbit of the determinant polynomial.  \\

The determinant $\Det := \{\Det_w\}_{w \in \N}$ is complete (under $p$-projections) for the class $\mathsf{VBP}$ \footnote{Class $\mathsf{VBP}$ consists of polynomial families that are computable by polynomial-size algebraic branching programs (ABP). ABP is a powerful model for computing polynomials that subsumes arithmetic formulas.}\cite{MahajanV97}. Likewise, the family of \emph{iterated matrix multiplication} polynomials is also complete for the class $\mathsf{VBP}$, and has been used quite a bit in proving arithmetic circuit lower bounds. In this sense, the two families have the same complexity\footnote{Consider a class $\CAL{C}$ of arithmetic circuits that is closed under affine projections, e.g., the class of depth three circuits. A super-polynomial lower bound for circuits in $\CAL{C}$ computing the determinant implies a super-polynomial lower bound for circuits in $\CAL{C}$ computing the iterated matrix multiplication polynomial (IMM) and vice versa. Thus, Det and IMM have the same complexity, and one may study the ``permanent versus IMM'' problem in the same vein as the permanent versus determinant problem. On the other hand, if $\CAL{C}$ is not closed under affine projections, then there are classes (like multilinear formulas) for which a super-polynomial lower bound is known for determinant  \cite{Raz09} but not for IMM.}. But, do they have similar equivalence testing complexity? Our work here, in conjunction with \cite{GargGKS19} and \cite{KayalNST19}, gives an answer to this question. \\

\textbf{Iterated matrix multiplication.} Two natural versions of the iterated matrix multiplication polynomial are: a) $\oIMM_{w,d}$ that is defined as the $(1,1)$ entry of the product of $d$ many $w\times w$ symbolic matrices (i.e., matrices whose entries are distinct variables), and b) $\IMM_{w,d}$ that is defined as the trace of the product of $d$ many $w\times w$ symbolic matrices. The $\oIMM := \{\oIMM_{w,d}\}_{w,d \in \N}$ family has been studied more from the lower bound perspective \cite{NisanW97,FournierLMS15,KumarS17, KayalNS20, KayalS15a, KayalST18, ChillaraLO19} because it naturally captures the algebraic branching program model (see Section \ref{subsection: abps and matrix products}). On the other hand, $\IMM := \{\IMM_{w,d}\}_{w,d \in \N}$ has been studied in \cite{Grochow12, Lan15, Gesmundo15, GIP17}\footnote{Actually, \cite{GIP17} studied a related polynomial $\text{Tr-Pow}_{w,d}$, which is the trace of the $d$-th power of a $w \times w$ symbolic matrix. They showed that a particular line of attack prescribed by GCT, namely \emph{orbit occurrence obstructions}, cannot prove super-linear lower bound on the ``Tr-Pow complexity'' of the permanent. We are not aware of a similar result (or, more generally, a result that rules out the \emph{occurrence obstructions} approach as in \cite{BurgisserIP16, IkenmeyerP16}) with Tr-Pow (or Det) replaced by Tr-IMM.} owing to its nice structural properties (pertaining to its group of symmetries and the associated Lie algebra) that may be quite useful for studying GCT methods when applied to the ``Permanent versus $\IMM$'' problem. $\oIMM$ and $\IMM$ are also complete for the class $\mathsf{VBP}$. Interestingly, the three polynomials -- $\Det_w$, $\oIMM_{w,d}$ and $\IMM_{w,d}$ -- are characterized by their respective groups of symmetries \cite{Frobenius97, KayalNST19, Gesmundo15}. \\

\textbf{Equivalence testing for iterated matrix multiplication.} How does equivalence testing for $\oIMM$ and $\IMM$ relate to that of $\Det$? In \cite{KayalNST19}, a randomized polynomial-time equivalence testing algorithm was given for $\oIMM$ over $\C,\Q$ and finite fields. Comparing this with the above-mentioned results on determinant equivalence test \cite{Kayal12, GargGKS19}, we see that the complexity of equivalence tests for $\Det$ and $\oIMM$ are quite different over $\Q$ (unless integer factoring is easy). Is this also the case between $\Det$ and $\IMM$? One may be tempted to say `yes' owing to the closeness of the definitions of $\oIMM$ and $\IMM$. However, contrary to this first impression, we show that equivalence testing for $\Det$ and that for $\IMM$ are randomized polynomial-time Turing reducible to each other over $\C$, $\Q$ and finite fields\footnote{The reduction works over any field $\F$ of characteristic zero or sufficiently large. We also require that univariate polynomial factoring over $\F$ can be done efficiently.} (see Corollary \ref{corollary: equivalence of det, trace and fmai}). Thus, viewed along this line, $\Det$ and $\IMM$ are closer to each other than to $\oIMM$.\footnote{Talking of the difference between the `trace model' and the `(1,1) model', a recent work \cite{BlaserIMPS20} showed that in the non-commutative setting, the border width complexity and the width complexity of a polynomial are \emph{not} always equal for the trace-ABP model, unlike the case for the classical $(1,1)$-ABP model \cite{Nisan91}.} For brevity, we would henceforth denote the equivalence testing problems for $\Det$ and $\IMM$ by $\EDET$ and $\EIMM$ respectively. \\

\textbf{Connections to algebra isomorphism and $3$-tensor isomorphism.} As mentioned before, cubic form equivalence, algebra isomorphism and $3$-tensor isomorphism are polynomial-time equivalent. Moreover, degree-$d$ form equivalence reduces to cubic form equivalence \cite{AgrawalS05, AgrawalS06} and $d$-tensor isomorphism reduces to $3$-tensor isomorphism \cite{GrowchowQ19} in polynomial-time, if $d$ is bounded. $\Det$ and $\IMM$ being two important polynomial families, we wonder if $\EDET$ and $\EIMM$ can be linked with any natural case of algebra isomorphism. Further, do $\EDET$ and $\EIMM$ reduce to any special case of cubic form equivalence or $3$-tensor isomorphism? We show that the answers to these are `yes'. The relevant problems are the \emph{full-matrix algebra isomorphism} ($\FMAI$) problem and the $3$-tensor isomorphism problem for the family of \emph{matrix multiplication tensors} ($\MMTI$). \\


$\FMAI$ is a well-studied problem in computer algebra which is defined as follows: Given a basis of a matrix algebra $\CAL{A} \subseteq \CAL{M}_m(\F)$, check if $\CAL{A}$ is isomorphic\footnote{i.e., isomorphic as algebras over $\F$.} to $\CAL{M}_w(\F)$, where $\CAL{M}_m(\F)$ is the algebra of $m\times m$ matrices over $\F$ and $\textnormal{dim}_{\F}(\CAL{A}) = w^2$; if yes, then output an isomorphism from $\CAL{A}$ to $\CAL{M}_w(\F)$. A randomized polynomial-time algorithm to solve $\FMAI$ over finite fields was given in \cite{Ronyai87,Ronyai90}, whereas over $\Q$ a randomized Turing reduction from $\FMAI$  to integer factoring was shown in \cite{IvanyosRS12,CremonaFOSS15}. The reduction is polynomial-time if $\textnormal{dim}_{\Q}(\CAL{A})$ is bounded. Also, \cite{BabaiR90, Eberly89} gave a randomized polynomial-time algorithm that outputs an isomorphism from $\CAL{A}\otimes_{\Q} \mathbb{L}$ to $\CAL{M}_w(\mathbb{L})$, where $\mathbb{L}$ is a degree $w$ extension field of $\Q$, if $\CAL{A}$ is isomorphic to $\CAL{M}_w(\Q)$. The decision version of $\FMAI$ over $\Q$ is in $\mathsf{NP} \cap \mathsf{coNP}$  \cite{Ronyai92}. The results for $\EDET$ in \cite{GargGKS19} were obtained by giving a randomized poly-time Turing reduction from $\EDET$ to $\FMAI$. In this work, we give a randomized polynomial-time Turing reduction from $\EIMM$ to $\EDET$ (Theorem \ref{theorem: reduction from equivalence testing of trace to det}). \\

A $d$-tensor is a degree-$d$ form (i.e., a degree-$d$ homogeneous polynomial) $f(\vecx_1, \vecx_2, \ldots, \vecx_d)$ whose every monomial has exactly one variable from each of the sets $\vecx_1, \vecx_2, \ldots, \vecx_d$. The $d$-tensor isomorphism problem is the following: Given two $d$-tensors $f(\vecx_1, \vecx_2, \ldots, \vecx_d)$ and $g(\vecx_1, \vecx_2, \ldots, \vecx_d)$ decide if there exist $A_1 \in \GL(|\vecx_1|, \F), \ldots, A_d \in \GL(|\vecx_d|, \F)$ such that $f = g(A_1 \vecx_1, A_2\vecx_2, \ldots, A_d \vecx_d)$. The $d$-tensor isomorphism problem for a family of $d$-tensors is defined accordingly, just like equivalence testing for a family of polynomials. $\MMTI$ is the $3$-tensor isomorphism problem for the family of matrix multiplication tensors $\{\IMM_{w,3}\}_{w \in \N}$. The matrix multiplication tensor $\IMM_{w,3}$ is a crucial object in the study of asymptotically fast algorithms for multiplying two $w \times w$ matrices. In this paper, we give a randomized polynomial-time Turing reduction from $\FMAI$ to $\MMTI$ (Theorem \ref{theorem: reduction from fmai to emi}). Further, it follows easily from the symmetries of $\IMM_{w,d}$ (\cite{Gesmundo15}, see Lemma \ref{fact: symmetries of IMM}) that $\MMTI$ reduces in polynomial-time to $\EIMM$. \\

Thus, the above results together with the reduction in \cite{GargGKS19} show that the four problems -- $\EIMM$, $\EDET$, $\FMAI$ and $\MMTI$ -- are randomized polynomial-time Turing reducible to each other. Although, the equivalence between $\MMTI$ and $\FMAI$ has the same essence as the equivalence between 3-tensor isomorphism (or cubic form equivalence) and algebra isomorphism, our proofs are quite different from the proofs in \cite{GrowchowQ19, FGS19, AgrawalS05, AgrawalS06}\footnote{The reductions in these prior works are deterministic and hold for the decision versions of the problems, whereas the reductions here are randomized and for the search versions of the problems.}. In particular, we do not see any easy adaptation of the arguments in \cite{GrowchowQ19, FGS19, AgrawalS05, AgrawalS06} leading to the results mentioned above. Our proofs link $\MMTI$ with $\FMAI$, via $\EIMM$ and $\EDET$, by exploiting the structure of the Lie algebra of $\IMM_{w,d}$ (which is in the same spirit as the reduction from $\EDET$ to $\FMAI$ in \cite{GargGKS19} using the Lie algebra of $\Det_{w}$). Also, the reduction from $d$-tensor isomorphism (similarly, degree-$d$ form equivalence) to $3$-tensor isomorphism (respectively, cubic form equivalence) in \cite{GrowchowQ19, AgrawalS05, AgrawalS06} is efficient only if $d$ is a constant. Whereas, our randomized reduction from testing equivalence to $\IMM_{w,d}$ to $\MMTI$ runs in time $\poly(w,d)$.

\subsection{The results (stated formally)}
The polynomial $\IMM_{w,d} := \text{tr}(Q_0\cdot Q_1\ldots Q_{d-1})$, where $Q_k$ is a $w\times w$ symbolic matrix in $\vecx_k$ variables. Throughout, we will assume that $w \geq 2$, $d \geq 3$ and $\char(\F) = 0$ or $> (w^2d)^5$, and univariate polynomial factoring over $\F$ can be done in probabilistic polynomial time. The restriction on the characteristic of $\F$ has not been optimized in this paper. 


\begin{theorem}[$\EIMM$ to $\EDET$]\label{theorem: reduction from equivalence testing of trace to det}
There is a randomized algorithm that takes as input black-box access to an $n$-variate degree-$d$ polynomial $f$ and oracle access to $\EDET$ over $\F$, and does the following with high probability: If there is a $w\in \N$ such that $f$ is equivalent to $\IMM_{w,d}$, then it outputs an $A \in \GL(n,\F)$ such that $f = \IMM_{w,d}(A\vecx)$; otherwise it outputs `No such $w$ exists'. The algorithm runs in $\poly(n, \beta)$ time, where $\beta$ is the bit length of the coefficients of $f$. 
\end{theorem}
The reduction is given in Section \ref{sec: reduction from IMM to DET}. Theorem \ref{theorem: reduction from equivalence testing of trace to det} implies a randomized poly-time algorithm for $\EIMM$ over $\C$ and finite fields, and also over $\Q$ (provided the algorithm has access to integer factoring oracle and $w$ is bounded) via known results on $\EDET$ \cite{Kayal12, GargGKS19}. Two other remarks: 
\begin{enumerate}
\item  \emph{No knowledge of $w$}: The algorithm requires no knowledge of $w$, if the input polynomial $f$ is equivalent to $\IMM_{w,d}$ for some $w\in \N$ then the algorithm finds such a $w$. 

\item \emph{Reduction to $\IMMTI$}: The \emph{tensor isomorphism} problem for $\IMM$ (denoted $\IMMTI$) is as follows: 
Given blackbox access to a $d$-tensor $g(\vecx_0, \ldots, \vecx_{d-1})$, check if there are $B_0, \ldots, B_{d-1} \in \GL(w^2, \F)$ such that $g= \IMM_{w,d}(B_0\vecx_0, \ldots, B_{d-1}\vecx_{d-1})$, and if yes then output such $B_0, \ldots, B_{d-1}$. 
The algorithm in Theorem \ref{theorem: reduction from equivalence testing of trace to det} first reduces $\EIMM$ to $\IMMTI$ (finding $w$ in this step), and then solves $\IMMTI$ using $\EDET$ oracle over $\F$. The reduction from $\EIMM$ to $\IMMTI$ (which resembles a similar reduction used in the equivalence test for $\oIMM$ \cite{KayalNST19}) does not require oracle access to $\EDET$. A randomized polynomial-time algorithm for $\IMMTI$ \emph{over $\C$} was given in \cite{Grochow12}, but the algorithm there does not reduce $\IMMTI$ to $\EDET$. 
\end{enumerate}

\begin{theorem}[$\FMAI$ to $\MMTI$] \label{theorem: reduction from fmai to emi}
There is a randomized algorithm that takes as input a basis of an algebra $\CAL{A} \subseteq \CAL{M}_m(\F)$, and oracle access to $\MMTI$, and does the following with high probability: If $\CAL{A} \cong \CAL{M}_w(\F)$, where $w^2 = \textnormal{dim}_{\F}(\CAL{A})$, then it outputs `Yes'; otherwise it outputs `No such $w\in \N$ exists'. If  the algorithm outputs `Yes', then it also outputs an algebra isomorphism from $\CAL{A}$ to $\CAL{M}_w(\F)$. The algorithm runs in $\poly(m, \beta)$ time, where $\beta$ is the bit length of the entries of the input basis matrices. 
\end{theorem} 
The algorithm is given in Section \ref{subsec: proof of reduction from fmai to imm}. It uses a characterization of $\IMM_{w,d}$ by the Lie algebra $\GIMM$ of its group of symmetries (Lemma \ref{lemma: characterization of trace by its lie algebra}) along with a nice choice of basis of $\GIMM$ (Section \ref{section: lie algebra of imm}) to reduce $\FMAI$ to degree four $\IMMTI$ in \emph{deterministic} polynomial time, which in turn reduces to $\MMTI$ in randomized polynomial time (Theorem \ref{theorem: reduction from equivalence testing for degree d to degree 3}). Two more remarks on Theorem \ref{theorem: reduction from fmai to emi}:
\begin{enumerate}

\item \emph{$\MMTI$ to $\EIMM$}: Using oracle access to $\EIMM$, it is easy to solve $\MMTI$ (in fact $\IMMTI$) in polynomial time: Since a polynomial identity test at the end of a $\IMMTI$ algorithm ensures that the output of the algorithm is correct, it suffices to prove that if the input to a $\EIMM$ algorithm is a $d$-tensor $f$ that is isomorphic to $\IMM_{w,d}$, then the algorithm outputs $d$ matrices $B_0,\ldots, B_{d-1}$ such that $f(\vecx) = \IMM_{w,d}(B_0\vecx_0, \ldots, B_{d-1}\vecx_{d-1})$. This is true as any algorithm for $\EIMM$ outputs a block-diagonal matrix $B$ such that $f(\vecx) = \IMM_{w,d}(B\vecx)$ (from Lemma \ref{fact: symmetries of IMM}). Matrices $B_0,\ldots, B_{d-1}$ can be easily derived from $B$. 

\item \emph{A reduction from $\FMAI$ to $\EDET$}: A Turing reduction from $\FMAI$ to $\EDET$ over $\F$ was given in \cite{GargGKS19} that runs in exponential time. We improve this run-time significantly: Theorems \ref{theorem: reduction from equivalence testing of trace to det} and \ref{theorem: reduction from fmai to emi} imply that $\FMAI$ is in fact randomized polynomial-time Turing reducible to $\EDET$.
\end{enumerate}

\begin{corollary}\label{corollary: equivalence of det, trace and fmai}
It follows from Theorems \ref{theorem: reduction from equivalence testing of trace to det} and \ref{theorem: reduction from fmai to emi}, and the randomized polynomial-time Turing reduction from $\EDET$ to $\FMAI$ in \cite{GargGKS19}, that the four problems -- $\EIMM$, $\EDET$, $\FMAI$ and $\MMTI$ -- are randomized polynomial-time equivalent under Turing reductions (see Figure \ref{figure: connections between problems} below). 
\end{corollary}
As mentioned before, the next theorem (proved in Appendix \ref{sec: reduction from degree d to degree 3}) is used in the proof of Theorem \ref{theorem: reduction from fmai to emi}. 
\begin{theorem}[$\IMMTI$ to $\MMTI$]\label{theorem: reduction from equivalence testing for degree d to degree 3}
There is a randomized algorithm that takes as input black-box access to an $n$-variate $d$-tensor $f(\vecx_0, \ldots, \vecx_{d-1})$, and oracle access to $\MMTI$, and does the following with high probability: If $f$ is isomorphic to $\IMM_{w,d}$, then it outputs $B_0, B_1, \ldots, B_{d-1} \in \GL(w^2,\F)$ such that $f = \IMM_{w,d}(B_0\vecx_0, \ldots, B_{d-1}\vecx_{d-1})$; otherwise it outputs `No'. The algorithm runs in $\poly(n, \beta)$ time, where $\beta$ is the bit length of the coefficients of $f$.
\end{theorem}

\input{secappendix_figure_connections.tex}
%

%% file: secappendix_figure_connections.tex
The figure below is a depiction of Corollary \ref{corollary: equivalence of det, trace and fmai}. An arrow from Problem A to B indicates a randomized polynomial-time Turing reduction from A to B.
\begin{figure}[h]
\centering
\input{figure_connections.tex}
\caption{Reductions between $\EIMM$, $\EDET$, $\FMAI$, and $\MMTI$}
\label{figure: connections between problems}
\end{figure}
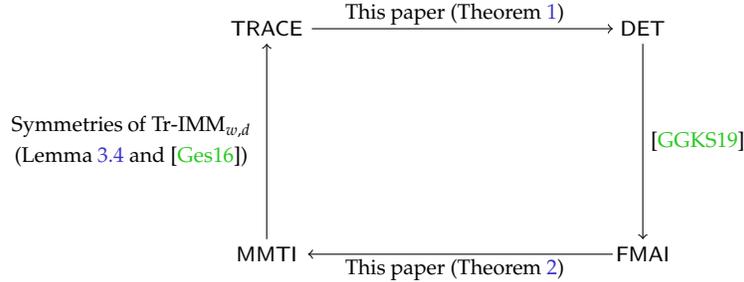

%% file: figure_connections.tex
\begin{tikzpicture}
\coordinate (a) at (0,0);
\coordinate (b) at ($(a) + (5,0)$);
\coordinate (c) at ($(a) + (5,-3)$);
\coordinate (d) at ($(a) + (0,-3)$);

\node at ($(a)$) [fill=white!100!] {\scriptsize $\EIMM$};
\node at ($(b)$) [fill=white!100!] {\scriptsize $\EDET$};
\node at ($(c)$) [fill=white!100!] {\scriptsize $\FMAI$};
\node at ($(d)$) [fill=white!100!] {\scriptsize $\MMTI$};

\node at ($(a)+ (2.5,0.2)$) [fill=white!100!] {\scriptsize This paper (Theorem \ref{theorem: reduction from equivalence testing of trace to det})};
\node at ($(b)+ (0.7,-1.5)$) [fill=white!100!] {\scriptsize \cite{GargGKS19}};
\node at ($(d)+ (2.5,-0.2)$) [fill=white!100!] {\scriptsize This paper (Theorem \ref{theorem: reduction from fmai to emi})};
\node at ($(a)+ (-1.8,-1.3)$) [fill=white!100!] {\scriptsize Symmetries of $\IMM_{w,d}$};
\node at ($(a)+ (-1.8,-1.7)$) [fill=white!100!] {\scriptsize (Lemma \ref{fact: symmetries of IMM} and \cite{Gesmundo15})};

\draw [->] ($(a) + (0.6,0)$) -- ($(b) - (0.4,0)$);
\draw [->] ($(b) + (0,-0.2)$) -- ($(c) + (0,0.2)$);
\draw [->] ($(c) - (0.4,0)$) -- ($(d) + (0.55,0)$);
\draw [->] ($(d) + (0,0.2)$) -- ($(a) + (0,-0.2)$);
\end{tikzpicture}

%% file: sec_prelim.tex
\section{Notations and definitions}\label{subsec: Variable ordering, notations and definitions}
Recall that $\IMM_{w,d} := \tr(Q_0\cdot Q_1\ldots Q_{d-1})$, where $Q_k = (x_{ij}^{(k)})_{i,j \in [w]}$. Let $\vecx_k = \{x_{ij}^{(k)}\}_{i,j \in [w]}$, $\vecx = \uplus_{k\in [0,d-1]} \vecx_k$, and $n = w^2d$. At times, we will refer to the $\vecx$ variables as $x_1, \ldots, x_n$. The $\vecx$ variables are ordered as $\mathbf{x}_0 > \mathbf{x}_1 > \ldots > \mathbf{x}_{d-1}$, and within a variable set $\mathbf{x}_k$, if $k$ is even (similarly, odd) then the variables are ordered in row-major (respectively, column-major) fashion. The rows and columns of a matrix in $\CAL{M}_n = \CAL{M}_n(\F)$, and the entries of a column vector in $\F^n$ are indexed by $\vecx$ variables ordered as above. A matrix in $\CAL{M}_n$ is called \emph{block-diagonal} if the row and column of every non-zero entry of the matrix is indexed by variables from the same variable set. A few more basic definitions and terminologies about matrices, matrix products and ABP are given in Appendix \ref{subsection: abps and matrix products}. The indices $k,\ell \in [0,d-1]$ will be treated as elements in $\Z/d\Z$, i.e., $k+1 = 0$ if $k=d-1$. 
Let $\CAL{L} \subseteq \CAL{M}_n$. A subspace $\CAL{U} \subseteq \F^n$ is \emph{$\CAL{L}$-invariant} if for all $M \in \CAL{L}$, $M\cdot \CAL{U} \subseteq \CAL{U}$. 
\begin{definition}[Irreducible invariant subspace]{\label{def:irr_inv_sub_spaces}}
An $\CAL{L}$-invariant subspace $\mathcal{U} \subseteq \F^n$ is irreducible if there are no proper $\CAL{L}$-invariant subspaces $\mathcal{U}_1$ and $\mathcal{U}_2$ of $\mathcal{U}$ such that $\mathcal{U} = \mathcal{U}_1 \oplus \mathcal{U}_2$. 
\end{definition}
\begin{definition}[Closure of a vector]
The closure of a vector $\mathbf{v} \in \mathbb{F}^n$ under the action of $\mathcal{L} \subseteq \CAL{M}_{n}$ is the smallest $\CAL{L}$-invariant subspace of $\F^n$ containing $\mathbf{v}$. 
\end{definition}
An algorithm to compute the closure of a vector in polynomial-time is given in \cite{KayalNST19}.
An easy-to-work-with definition of the Lie algebra of the group of symmetries of a polynomial was given in \cite{Kayal12}. For brevity, we will call it the Lie algebra of a polynomial.\footnote{Geometrically speaking, the Lie algebra of an $n$-variate polynomial $f(\vecx)$ is the subspace of $\CAL{M}_n(\F)$ obtained by translating the tangent of the algebraic set $\{A \in \CAL{M}_n~:~ f(A\vecx) = f(\vecx)\}$ at $A = I_n$ and making it pass through origin.}
\begin{definition}[Lie algebra $\mathfrak{g}_f$ of a polynomial $f$]{\label{def:lie_algebra}}
The Lie algebra of an $n$-variate polynomial $f(\vecx)$ is denoted as $\mathfrak{g}_f$ and it consists of matrices $E = (e_{ij})_{i,j \in [n]} \in \CAL{M}_n$ that satisfy $ \sum_{i,j \in [n]}e_{ij} x_{j}\cdot\frac{\partial f}{\partial x_i} = 0$.
\end{definition}
Note that $\G_f$ is a vector space. It also follows that a basis of $\G_f$ can be computed in randomized polynomial-time from blackbox access to $f$ by solving a linear system (see \cite{Kayal12}). 
\begin{fact}\label{fact: lie algebra conjugacy}
If $f(\vecx) = g(A\vecx)$ for an $A \in \GL(n,\F)$, then $\G_f = A^{-1}\G_g A$.  
\end{fact}
%

%% file: sec_Lie_algebra_trace.tex
\section{Symmetries and Lie algebra of $\IMM$}\label{section: lie algebra of imm}
The symmetries and the Lie algebra $\GIMM$ of $\IMM_{w,d}$ have been studied in \cite{Gesmundo15} over $\C$. Here, we work out the exact structure of the matrices in $\GIMM$ with respect to the variable ordering mentioned above, and use it to identify the $\GIMM$-invariant subspaces of $\F^n$ and the symmetries of $\IMM_{w,d}$ over $\F$. These facts about the Lie algebra and the symmetries will be used in the proofs of Theorems \ref{theorem: reduction from equivalence testing of trace to det}, \ref{theorem: reduction from fmai to emi} and \ref{theorem: reduction from equivalence testing for degree d to degree 3}. The missing proofs of this section are given in Appendix \ref{secappendix: lie algebra of imm}.

\begin{claim}{\label{claim: gtrimm is block diagonal}}
If $E \in \G_{\IMM}$ then $E$ is block-diagonal.
\end{claim}

Define the spaces $\CAL{B}_0, \ldots , \CAL{B}_{d-1}$ of block-diagonal matrices as follows: Every matrix in $\CAL{B}_k$ is a block-diagonal matrix whose non-zero entries are confined to the rows and columns indexed by $\vecx_k$ and $\vecx_{k+1}$ variables. For $k \in [0,d-2]$ and $B \in \CAL{B}_k$, let $[B]_k$ be the $2w^2 \times 2w^2$ sub-matrix of $B$ whose rows and columns are indexed by $\vecx_k$ and $\vecx_{k+1}$ variables. For $B \in \CAL{B}_{d-1}$, let $[B]_{d-1}$ be the $2w^2 \times 2w^2$ sub-matrix of $B$ whose rows and columns are indexed by $\vecx_{d-1}$ and $\vecx_0$ variables, i.e., we let the $\vecx_{d-1}$ variables index the rows and columns of $B_{d-1}$ before the $\vecx_0$ variables. If $d$ is even then
\begin{eqnarray} \label{equation: lie algebra matrices}
\CAL{B}_k &:=& \left\{ B \in \CAL{M}_n ~:~ [B]_k =  \begin{bmatrix} I_w \otimes M^T & \mathbf{0} \\ \mathbf{0}  & -I_w \otimes M \end{bmatrix}  ~\text{for}~ M \in \CAL{M}_w \right\} ~~\text{if $k$ is even}, \nonumber \\
&:=& \left\{ B \in \CAL{M}_n ~:~ [B]_k =  \begin{bmatrix} M^T \otimes I_w & \mathbf{0} \\ \mathbf{0}  & -M \otimes I_w \end{bmatrix}  ~\text{for}~ M \in \CAL{M}_w \right\} ~~\text{if $k$ is odd}.
\end{eqnarray}
If $d$ is odd, then the definition of $\CAL{B}_k$ remains the same except for $\CAL{B}_{d-1}$ which is defined as
\begin{equation*}
\CAL{B}_{d-1} := \left\{ B \in \CAL{M}_n ~:~ [B]_{d-1} =  \begin{bmatrix} I_w \otimes M^T & \mathbf{0} \\ \mathbf{0}  & -M \otimes I_w \end{bmatrix}  ~\text{for}~ M \in \CAL{M}_w \right\}.
\end{equation*}

\begin{lemma}\label{theorem: lie algebra of of trace}
The space $\CAL{B}_0 + \ldots + \CAL{B}_{d-1}$ is contained in $\GIMM$. 
\end{lemma}
\begin{lemma}\label{lemma:block-diagonal matrix with two non-zero blocks in GIMM is in B_k}
Suppose $B\in \GIMM$ and there is a $k\in [0,d-1]$ such that the non-zero entries of $B$ are confined to the rows and columns that are indexed by $\vecx_k$ and $\vecx_{k+1}$ variables. Then  $B\in \CAL{B}_k$.
\end{lemma}
In fact $\GIMM = \CAL{B}_0 + \ldots + \CAL{B}_{d-1}$, however we do not prove this stronger statement here. Let $e_i \in \F^n$ be the vector with $1$ in the entry indexed by $x_i \in \vecx$ and zero elsewhere. A subspace of $\F^n$ is a coordinate subspace if it is spanned by a set of $e_i$'s. Let $\CAL{U}_k = \spa_{\F}\{e_i \,\mid \, x_i \in \vecx_k\}$. 
\begin{claim}{\label{claim:invar_subspace_is_coordinate}}
Any non-zero $\GIMM$-invariant subspace is a coordinate subspace of $\F^n$.
\end{claim} 
\begin{lemma}\label{lemma:irreducible_space_limm}
The only irreducible $\GIMM$-invariant subspaces of $\F^n$ are $\mathcal{U}_0, \ldots, \mathcal{U}_{d-1}$. 
\end{lemma}
\begin{corollary}\label{corollary:irr_inv_sub_f}
If $f = \IMM_{w,d}(A\vecx)$, where $A \in \GL(n,\F)$, then the only irreducible $\G_f$-invariant subspaces of $\F^n$ are $A^{-1}\CAL{U}_0, \ldots, A^{-1}\CAL{U}_{d-1}$. 
\end{corollary}
The above lemmas help us derive the group of symmetries of $\IMM_{w,d}$ over $\F$.

\begin{lemma}\label{fact: symmetries of IMM}
Let $\IMM_{w,d} = \tr(Q'_0\cdots Q'_{d-1})$, where $Q'_0\cdots Q'_{d-1} $ is a full-rank $(w,d,n)$-matrix product in $\vecx$ variables over $\F$. Then there are $C_0, \ldots, C_{d-1} \in \GL(w,\F)$ and $\ell\in [0,d-1]$ such that either $Q'_k = C_k\cdot Q_{\ell+k} \cdot C_{k+1}^{-1}$ for $k \in [0,d-1]$ ~or~ $Q'_k = C_k\cdot Q_{\ell -k}^T \cdot C_{k+1}^{-1}$ for $k \in [0,d-1]$.
\end{lemma}

%

%% file: sec_reduction_to_det_eq.tex
\section{Reduction from $\EIMM$ to $\EDET$:~ Proof of Theorem \ref{theorem: reduction from equivalence testing of trace to det}}\label{sec: reduction from IMM to DET}
The reduction is given in Algorithm \ref{algorithm: equivalence testing for trimm}. The algorithm proceeds by assuming that the input polynomial $f$ is equivalent to $\IMM_{w,d}$ for some $w\geq 2$. A final polynomial identity test (PIT) takes care of the case when it is not. Algorithm \ref{algorithm: equivalence testing for trimm} has two main steps -- reduction from $\EIMM$ to $\IMMTI$ (Algorithm \ref{algorithm: reduction from equivalence to multilinear equivalence for trimm}), and reduction from $\IMMTI$ to $\EDET$ (Algorithm \ref{algorithm: multilinear equivalence test for trimm}).  Algorithm \ref{algorithm: reduction from equivalence to multilinear equivalence for trimm} is inspired by a similar reduction in \cite{KayalNST19} for the $\oIMM$ polynomial. Below we discuss the proof strategy of Algorithm \ref{algorithm: reduction from equivalence to multilinear equivalence for trimm}, and give the details in Appendix \ref{secappendix: reduction to multilinear equivalence testing}. Algorithm \ref{algorithm: multilinear equivalence test for trimm} is given in Section \ref{section: multilinear equivalence testing}. \\ 

\textbf{Reduction from $\EIMM$ to $\IMMTI$.} First, we compute bases of the irreducible $\G_f$-invariant subspaces of $\F^n$. By Corollary \ref{corollary:irr_inv_sub_f}, these are bases of the spaces $A^{-1}\CAL{U}_{\sigma(0)}, \ldots, A^{-1}\CAL{U}_{\sigma(d-1)}$, where $\sigma$ is an unknown permutation on $\{0, \ldots, d-1\}$. As $\dim_{\F}(\CAL{U}_k) = w^2$, we get $w$. Now, let $V_k$ be the $n \times w^2$ matrix consisting of the basis vectors of $A^{-1}\CAL{U}_{\sigma(k)}$. Form the $n \times n$ matrix $V = [V_0 \mid V_1 \mid \ldots \mid V_{d-1}]$. Observe that $V = A^{-1} \cdot E$, where $E$ is a "block-permuted" invertible matrix (by the definition of $\CAL{U}_k$). Thus, $h(\vecx) := f(V\vecx) = \IMM_{w,d}(E\vecx)$. We now make use of the evaluation dimension measure (Definition \ref{definition: evaldim}) on $h$ to essentially ensure that $E$ is a block-diagonal matrix.  \\

\begin{algorithm}
\caption{Reduction from $\EIMM$ to $\EDET$} \label{algorithm: equivalence testing for trimm}
\begin{algorithmic}
\State INPUT: Blackbox access to an $n$-variate, degree $d$ polynomial $f$ and oracle access to $\EDET$.
\State OUTPUT: If there is an $w \in \mathbb{N}$  such that $f$ is equivalent to $\IMM_{w,d}$ then output an $A \in \GLNF$ such that $f(\vecx) = \IMM_{w,d}(A\vecx)$. Otherwise output `No such $w$ exists'.
\end{algorithmic}
\begin{algorithmic}[1]
\Statex \begin{center}\textcolor{gray}{Reduction to $\IMMTI$} \end{center}
\State Use Algorithm~\ref{algorithm: reduction from equivalence to multilinear equivalence for trimm} with input $f$ to compute $A' \in \GLNF$ and a $w\in \N$ such that $h(\vecx) = f(A'\vecx)$ is a $d$-tensor in the variable sets $\vecx_0, \ldots, \vecx_{d-1}$ which is isomorphic to $\IMM_{w,d}$. If Algorithm \ref{algorithm: reduction from equivalence to multilinear equivalence for trimm} outputs 'No', output `No such $w$ exists'. 

~
\Statex \begin{center}\textcolor{gray}{Reduction from $\IMMTI$ to $\EDET$} \end{center}
\State Use Algorithm~\ref{algorithm: multilinear equivalence test for trimm} with input $h$, $w$ and oracle access to $\EDET$ to compute matrices $B_0, \ldots ,B_{d-1} \in \GL(w^2, \F)$ such that $h(\vecx) = \IMM_{w,d}(B_0\vecx_0, \ldots, B_{d-1}\vecx_{d-1})$. If Algorithm \ref{algorithm: multilinear equivalence test for trimm} outputs 'No' then output `No such $w$ exists'. 
\State Let $B\in \GL(n,\F)$ be the block-diagonal matrix whose $k$-th block is $B_k$, and let $A = B(A')^{-1}$. 

~
\Statex \begin{center} \textcolor{gray}{Final PIT} \end{center}
\State Pick a random point $\mathbf{a} \in S^n$ where $S \subseteq \F$ is of size $n^5$. If $f(\mathbf{a}) = \IMM_{w,d}(A\mathbf{a})$ then output $w$ and $A$, else output `No such $w$ exists'.
\end{algorithmic}
\end{algorithm}
\input{subsec_multilinear_equivalence_testing}

%

%% file: subsec_multilinear_equivalence_testing.tex
\subsection{Reduction from $\IMMTI$ to $\EDET$}{\label{section: multilinear equivalence testing}}

We will use a few terminologies and notations about matrices, matrix products and ABP that are defined in Appendix \ref{subsection: abps and matrix products}. The following two claims (proved in Appendix \ref{secappendix: reduction to det}) help in the argument. 
\begin{claim}{\label{claim: uniqueness from transpose}}
Let $X$ be a $w\times w$ full-rank linear matrix and $Y=I_w\otimes X$. Then there does not exist non-zero matrices $T,S \in \CAL{M}_{w^2}(\F)$ such that $T\cdot Y = Y^T\cdot S$.
\end{claim}
\begin{claim}{\label{claim: Uniqueness of commuting matrices}}
Let $X$ be a $w \times w$ full-rank linear matrix and $Y = I_w\otimes X$, and suppose $T,S \in \CAL{M}_{w^2}(\F)$ such that $T\cdot Y = Y \cdot S$. Then $T = S = M \otimes I_w$ for some $M \in \CAL{M}_w(\F)$.
\end{claim}
The correctness of Algorithm \ref{algorithm: multilinear equivalence test for trimm} is argued below by tracing its steps. \vspace{0.1in} 

\begin{algorithm}[ht!]
\caption{Reduction from $\IMMTI$ to $\EDET$} \label{algorithm: multilinear equivalence test for trimm}
\begin{algorithmic}
\State INPUT: A $w\in \N$, blackbox access to $d$-tensor $h(\vecx_0, \ldots, \vecx_{d-1})$ that is isomorphic to $\IMM_{w,d}$, and oracle access to $\EDET$. 
\State OUTPUT: Matrices $B_0, \ldots , B_{d-1} \in \GL(w^2,\F)$ such that $h(\vecx) = \IMM_{w,d}(B_0\vecx_0, \ldots , B_{d-1}\vecx_{d-1})$.
\end{algorithmic}
\begin{algorithmic}[1]
\State Use the set-multilinear ABP reconstruction algorithm (which follows from \cite{KlivansS03}) to construct a $(w^2,d,n)$ set-multilinear ABP $Y'_0\ldots Y'_{d-1}$ in $\vecx_0, \ldots ,\vecx_{d-1}$ variables that computes $h$. 
\State For $k\in [1,d-2]$, use the factorization algorithm in \cite{KaltofenT90} to compute blackbox access to a degree-$w$ polynomial $g_k$ such that $\det(Y'_k) = \alpha_k g_k(\vecx_k)^w$, where $\alpha_k\in \F^{\times}$.
\State For $k\in [1,d-2]$, use the $\EDET$ oracle on input $g_k$ to compute $X'_k$ such that $\det(X'_k) =g_k$. If $\EDET$ returns $g_k$ is not equivalent to $\Det_w$, then output `No'.

~
\State For $k \in [1,d-2]$, let $Z_k = I_w \otimes X'_k$.
\State For $k \in [1,d-2]$, compute $T'_{k-1},~ S'_{k} \in \GL(w^2, \F)$ such that either $T'_{k-1} \cdot Y'_{k} = Z_k \cdot S'_k$ or $T'_{k-1} \cdot Y'_{k} = Z_k^{T} \cdot S'_k$. If both equalities are satisfied, output `No' (see Observation \ref{observation:exactly_one_succeds}).

~
\State Let $\widehat{Y}_0 = Y'_0\cdot (T'_0)^{-1}$,~ $\widehat{Y}_k = (T'_{k-1})\cdot Y'_k\cdot (T'_k)^{-1}$ \ for $k\in [1,d-3]$,~ $\widehat{Y}_{d-2} = (T'_{d-3})\cdot Y'_{d-2}\cdot (S'_{d-2})^{-1}$,~  and $\widehat{Y}_{d-1} = S'_{d-2}\cdot Y'_{d-1}$. 
\State Let $\widehat{X}_{d-2}$ be such that $\widehat{Y}_{d-2} = I_w \otimes \widehat{X}_{d-2}$, and for $k\in [1,d-3]$ construct $\widehat{M}_k \in \GL(w, \F)$ and $\widehat{X}_k$ such that $\widehat{Y}_{k} = (\widehat{M}_k\otimes I_w) \cdot (I_w\otimes \widehat{X}_k)$.~~~ (See Observation \ref{obs:structure of widehatY ABP}.)
\State Let $\overline{Y}_{d-1} = (\prod_{k=1}^{d-3}(\widehat{M}_k\otimes I_w))\cdot \widehat{Y}_{d-1}$. Construct $\widehat{X}_{d-1}$ such that its $(i,j)$-th entry is the $((j-1)w + i)$-th entry of $\overline{Y}_{d-1}$, and $\widehat{X}_0$ such that its $(i,j)$-th entry is the $((i-1)w + j)$-th entry of $\widehat{Y}_{0}$.  

~
\State Obtain the transformations $B_0,\ldots ,B_{d-1} \in \GL(w^2,\F)$ from (the entries of) $\widehat{X}_0, \ldots,  \widehat{X}_{d-1}$ respectively. Return $B_0,\ldots ,B_{d-1}$. 
\end{algorithmic}
\end{algorithm}
\textbf{Steps 1--3}: Assume that $h$ is isomorphic to $\IMM_{w,d}$. Hence, there is a full-rank $(w,d,n)$ set-multilinear matrix product $X_0\ldots X_{d-1}$ in $\vecx_0, \ldots, \vecx_{d-1}$ variables such that $h = \tr(X_0\ldots X_{d-1})$. From Fact \ref{fact: ABP computing IMM}, $h$ is computed by the $(w^2,d,n)$-set-multilinear ABP $Y_0\ldots Y_{d-1}$, where 
\begin{align*}
Y_0 & = (X_0(1,1), \ldots, X_0(1,w), X_0(2,1),\ldots, X_0(2,w), \ldots, X_0(w,1), \ldots, X_0(w,w)) \\
Y_{k} & = I_w\otimes X_k  ~~~~~\textnormal{for } k\in [1,d-2] \\
Y_{d-1} & = (X_{d-1}(1,1), \ldots , X_{d-1}(w,1), X_{d-1}(1,2),\ldots , X_{d-1}(w,2), \ldots, X_{d-1}(1,w), \ldots, X_{d-1}(w,w))^T ~.
\end{align*}
Using the randomized polynomial-time set-multilinear ABP reconstruction algorithm in \cite{KlivansS03}, a $(w^2,d,n)$ set-multilinear ABP $Y'_0\ldots Y'_{d-1}$ computing $h$ is constructed in Step $1$. It follows from the properties of this algorithm and the ABP $Y_0 \ldots Y_{d-1}$ that there are $T_{0}, \ldots ,T_{d-2}\in \GL(w^2,\F)$ so that
$$Y'_0 = Y_0 \cdot T_0~,~~~~~~~Y'_k = T_{k-1}^{-1}\cdot Y_k\cdot T_{k}  ~~~\textnormal{for } k\in [1,d-2], \text{ and }~~~~~~~Y'_{d-1} = T_{d-2}^{-1}\cdot Y_{d-1}~.\footnote{See Appendix \ref{subsection: abps and matrix products}: Set-multilinear ABP reconstruction, for an explanation.}$$
Hence, for all $k \in [1,d-2]$, $\det(Y'_k) =  c_k(\det(X_k))^w$, where $c_k \in \F^{\times}$. As the determinant polynomial is irreducible, at Step $2$, we have $g_k = \beta_k \det(X_k) = \det(\diag(\beta_k, 1, \ldots,1) \cdot X_k)$ for some $\beta_k \in \F^{\times}$ which implies $g_k$ is equivalent to $\Det_w$.
At step 3, $\EDET$ on input $g_k$ returns $X'_k$ such that  
\begin{equation*}
X_k = C_k\cdot X'_k \cdot D_k~~~~\text{or}~~~~X_k = C_k\cdot (X'_k)^T \cdot D_k ~~~~~~\text{where $C_k, D_k \in \GL(w,\F)$.}
\end{equation*}
 The above follows from the group of symmetries of $\Det_w$ (see Fact 1 in \cite{KayalNS19}). \vspace{0.1in}

\textbf{Steps 4--5}: At Step 4, for $k\in [1,d-2]$, the matrix $Z_k= I_w \otimes X'_k$ satisfies  
$$Y_k = (I_w\otimes C_k)\cdot Z_k\cdot (I_w\otimes D_k)~~~~\text{or}~~~~Y_k = (I_w\otimes C_k)\cdot Z_k^T\cdot (I_w\otimes D_k).$$
Hence, at Step 5 there are $T'_{k-1} := (I_w\otimes C_k^{-1})\cdot T_{k-1}$ and $S'_k := (I_w\otimes D_k)\cdot T_k$ in $\GL(w^2, \F)$ such that 
$$T'_{k-1} \cdot Y'_{k} = Z_k \cdot S'_k  ~~~~\text{or}~~~~T'_{k-1} \cdot Y'_{k} = Z_k^{T} \cdot S'_k.$$  
Observation \ref{observation:exactly_one_succeds} uses Claim \ref{claim: uniqueness from transpose} to show that at Step 5 we can identify between the above two cases, as only one of them is true (proof in Appendix \ref{secappendix: reduction to det}).
\begin{observation}{\label{observation:exactly_one_succeds}}
If $h(\vecx_0, \ldots, \vecx_{d-1})$ is isomorphic to $\IMM_{w,d}$ then for matrices $Y'_k$ and $Z_k$ as computed in Algorithm \ref{algorithm: multilinear equivalence test for trimm}, where $k \in [1,d-2]$, there are no matrices $T'_{k-1}, S'_k\in \GL(w^2,\F)$ such that both $T'_{k-1} \cdot Y'_{k} = Z_k \cdot S'_k$ and $T'_{k-1} \cdot Y'_{k} = Z_k^{T} \cdot S'_k$ are simultaneously true.
\end{observation}
At step 5 the matrices $T'_{k-1}$ and $S'_k$ are computed by solving linear equations. Choosing a solution at random from the solution space ensures that the computed matrices $T'_{k-1}$ and $S'_k$ are invertible with high probability. Henceforth, we assume that $T'_{k-1} \cdot Y'_{k} = Z_k \cdot S'_k$. The proof for $T'_{k-1} \cdot Y'_{k} = Z_k^{T} \cdot S'_k$ is similar. In Observation \ref{observation:uniquess_ti_si} we show that $T'_{k-1}$ and $S'_k$ are related to $T_{k-1}$ and $T_k$ respectively for $k\in [1,d-2]$. The proof of Observation \ref{observation:uniquess_ti_si}, which uses Claim \ref{claim: Uniqueness of commuting matrices}, is in Appendix \ref{secappendix: reduction to det}.
\begin{observation}[Uniqueness of $T'_{k-1}$ and $S'_k$]{\label{observation:uniquess_ti_si}} The matrices $T'_{k-1}$ and $S'_k$ computed at Step 5 of Algorithm \ref{algorithm: multilinear equivalence test for trimm}, where $k\in [1,d-2]$, satisfy the following: $(T'_{k-1})^{-1} = T_{k-1}^{-1} \cdot (I_w \otimes C_k)  \cdot (M_k^{-1} \otimes I_w)$ and $S'_k = (M_k \otimes I_w) \cdot (I_w \otimes D_k) \cdot T_k$, where $M_k \in \GL(w, \F)$. 
\end{observation}
\textbf{Steps 6--8}: Observation \ref{obs:structure of widehatY ABP} proved in Appendix \ref{secappendix: reduction to det} describes the structure of the matrices $\widehat{Y}_0,\ldots, \widehat{Y}_{d-1}$ computed at Step 6. Clearly, $\widehat{Y}_0\ldots \widehat{Y}_{d-1} = Y'_0\ldots Y'_{d-1}$ is a set-multilinear ABP computing $h$. 
\begin{observation}{\label{obs:structure of widehatY ABP}} Let $M_{1}, \ldots, M_{d-2}$ be the matrices as defined in Observation \ref{observation:uniquess_ti_si}. Then 
\begin{enumerate}
\item  $\widehat{Y}_k = (M_k M_{k+1}^{-1} \otimes I_w)\cdot(I_w \otimes (C_k^{-1} \cdot X_k \cdot C_{k+1}))$ ~~for $k \in [1,d-3]$, 
\item  $\widehat{Y}_{d-2} = I_w \otimes (C_{d-2}^{-1} \cdot X_{d-2} \cdot D_{d-2}^{-1})$, 
\item  $\widehat{Y}_0 = Y_0\cdot (I_w \otimes C_1)  \cdot (M_1^{-1} \otimes I_w)$, and $\widehat{Y}_{d-1} = (M_{d-2} \otimes I_w) \cdot (I_w \otimes D_{d-2})\cdot Y_{d-1}$.
\end{enumerate}
\end{observation}
By the above observation, at Step $7$, $\widehat{X}_{d-2} = C_{d-2}^{-1} \cdot X_{d-2} \cdot D_{d-2}^{-1}$. Moreover, the structure of $\widehat{Y}_k$ (as stated in the observation) enables the algorithm to factor it in Step $7$ and obtain $\widehat{X}_k, \widehat{M}_k$ such that 
$$\widehat{X}_k = a_k(C_k^{-1} \cdot X_k \cdot C_{k+1})~~~\text{and }\widehat{M}_k = a_k^{-1} (M_k\cdot M_{k+1}^{-1})~~~~~\text{for some } a_k\in \F^{\times}.$$ 
Let $a = \prod_{k=1}^{d-3}a_k$. Then at step 8,~ $\overline{Y}_{d-1} = a^{-1}\cdot (M_1\otimes I_w)\cdot (I_w\otimes D_{d-2})\cdot Y_{d-1}$. Now, it is a simple exercise to verify that at step 8
$$\widehat{X}_0 = (M_1^T)^{-1}\cdot X_0\cdot C_1 ~~~~\text{and}~~~~ \widehat{X}_{d-1} = a^{-1}(D_{d-2}\cdot X_{d-1}\cdot M_1^T). $$
\textbf{Step 9}: Therefore, $h = \tr(\widehat{X}_{0}\ldots \widehat{X}_{d-1})$. The transformation $B_k\in \GL(w^2,\F)$ is such that its rows are the coefficient vectors of the linear forms in $\widehat{X}_k$. Hence, $h = \IMM_{w,d}(B_0\vecx_0, \ldots, B_{d-1}\vecx_{d-1})$.

{}

%% file: sec_reduction_from_fmai_to_trace_equiv.tex
\section{Reduction from $\FMAI$ to $\MMTI$ :~ Proof of Theorem \ref{theorem: reduction from fmai to emi}}\label{section: reduction from fmai to imm}
\subsection{Characterization of $\IMM$ by its Lie algebra}\label{subsec: lie algebra characterization of imm}
The following lemma gives a characterization of $\IMM_{w,d}$ by its Lie algebra. The spaces $\CAL{B}_0, \ldots, \CAL{B}_{d-1}$ are as defined in Section \ref{section: lie algebra of imm}. The missing proofs are in Appendix \ref{secappendix: reduction from fmai to imm} . 
\begin{lemma}\label{lemma: characterization of trace by its lie algebra}
Let $f$ be a non-zero $d$-tensor in the variable sets $\vecx_0, \ldots ,\vecx_{d-1}$ such that for all $k \in [0,d-1]$ $\CAL{B}_k \subseteq \G_f$. Then there is an $\alpha \in \F^{\times}$ such that $f(\vecx) = \alpha \cdot \IMM_{w,d}(\vecx)$.
\end{lemma}
\begin{corollary}\label{corollary: lie algebra characterization for equivalent polynomials}
Let $B \in \GL(n,\F)$ be a block-diagonal matrix with individual blocks $B_0, \ldots, B_{d-1}$ and  $f$ be a non-zero $d$-tensor in the variable sets $\vecx_0, \ldots ,\vecx_{d-1}$ such that for all $k \in [0,d-1]$, $B^{-1}\cdot \CAL{B}_k \cdot B \subseteq \G_f$. Then there is an $\alpha \in \F^{\times}$ such that $f(\vecx) = \alpha \cdot \IMM_{w,d}(B_0\vecx_0,\ldots,$  $B_{d-1}\vecx_{d-1})$.
\end{corollary}

\subsection{Proof of Theorem \ref{theorem: reduction from fmai to emi}}\label{subsec: proof of reduction from fmai to imm}
Algorithm \ref{algorithm: reduction from fmai to trace equivalence} takes as input a basis  $\{ E_1,E_2, \ldots ,E_r \}$ of an algebra $\CAL{A} \subseteq \CAL{M}_m(\F)$, and if $\CAL{A} \cong \CAL{M}_w$ for some $w\in \N$, then it computes a $4$-tensor $f$ in the variable sets $\vecx_0, \vecx_{1}, \vecx_{2}, \vecx_3$ in deterministic polynomial time such that $f$ is isomorphic to $\IMM_{w,4}$. It then uses Algorithm \ref{algorithm: reduction from degree d to degree 3} in Theorem \ref{theorem: reduction from equivalence testing for degree d to degree 3} (see Appendix \ref{sec: reduction from degree d to degree 3}) to find an isomorphism from $f$ to $\IMM_{w,4}$ using oracle access to $\MMTI$ in randomized polynomial time. An easy check at the end of the algorithm ensures that if the algorithm outputs an isomorphism then it is correct. Thus, we need to prove that if $\CAL{A}$ is isomorphic to $\CAL{M}_w$ for some $w \in \N$ then the algorithm outputs an isomorphism. This is argued by tracing the steps of the algorithm assuming $\CAL{A}$ is isomorphic to $\CAL{M}_w$ for some $w \in \N$.
\begin{algorithm}[ht!]
\caption{Reduction from $\FMAI$ to $\MMTI$} \label{algorithm: reduction from fmai to trace equivalence}
\begin{algorithmic}
\State INPUT: A basis $\{ E_1,E_2, \ldots ,E_r \}$ of an algebra $\CAL{A} \subseteq \CAL{M}_m(\F)$, and oracle access to $\MMTI$.
\State OUTPUT: If $\CAL{A} \cong \CAL{M}_{w}(\F)$ for some $w \in \N$ then output an algebra isomorphism $\phi: \CAL{A} \rightarrow \CAL{M}_w$, otherwise output `No $w \in \N$ such that $\CAL{A} \cong \CAL{M}_w$'.
\end{algorithmic}
\begin{algorithmic}[1]
\State If $r \neq w^2$ for any $w \in \N$, then output `No $w \in \N$ such that $\CAL{A} \cong \CAL{M}_w$'.

\State Rename and order the basis elements as $E_{1,1}, \ldots ,E_{1,w}, \ldots, E_{w,1}, \ldots, E_{w,w}$. Compute matrices $L_{1,1}, \ldots ,L_{w,w}$, whose rows and columns are indexed by the above basis elements in order, as follows: $L_{i,j}$ is the matrix corresponding to the left multiplication of $E_{i,j}$ on $E_{1,1}, \ldots E_{w,w}$. In particular, $E_{i,j}\cdot E_{i_2,j_2} = \sum_{i_1,j_1\in [w]} L_{i,j}((i_1,j_1),(i_2,j_2)) E_{i_1,j_1}$.

~
\State Compute a basis of the space spanned by matrices in $\CAL{M}_{w^2}$ that commute with $\{ L_{1,1}^T, \ldots , L_{w,w}^T\}$. If the dimension of this space is not $w^2$, then output 'No $w \in \N$ such that $\CAL{A} \cong \CAL{M}_w$'. Otherwise, let the computed basis be $\{N_{1,1}, \ldots ,N_{w,w}\}$.

~
\State Compute a non-zero $4$-tensor $f$ in $\vecx_0, \ldots ,\vecx_{3}$ variables whose coefficients satisfy the following equations: a) for all $k \in [0,3]$, $k$ even, and for all $L \in \{L_{1,1}, \ldots L_{w,w}\}$
\begin{equation}\label{equation: for k odd}
    \sum_{i_1,j_1,i_2,j_2 \in [w^2]} L^T((i_1,j_1)(i_2,j_2)) x_{i_2,j_2}^{(k)} \frac{\partial f}{x_{i_1,j_1}^{(k)}} ~~~- \sum_{i_1,j_1,i_2,j_2 \in [w^2]} L((i_1,j_1)(i_2,j_2)) x_{j_2,i_2}^{(k+1)} \frac{\partial f}{x_{j_1,i_1}^{(k+1)}} =~~ 0 .
\end{equation}
b) for all $k \in [0,3]$, $k$ odd, and for all $N \in \{N_{1,1}, \ldots N_{w,w}\}$
\begin{equation}\label{equation: for k even}
    \sum_{i_1,j_1,i_2,j_2 \in [w^2]} N^T((i_1,j_1)(i_2,j_2)) x_{j_2,i_2}^{(k)} \frac{\partial f}{x_{j_1,i_1}^{(k)}} ~~~- \sum_{i_1,j_1,i_2,j_2 \in [w^2]} N((i_1,j_1)(i_2,j_2)) x_{i_2,j_2}^{(k+1)} \frac{\partial f}{x_{i_1,j_1}^{(k+1)}} = 0.
\end{equation}

~
\State Use Algorithm \ref{algorithm: reduction from degree d to degree 3} on input $f$ and with oracle access to $\MMTI$. If the algorithm outputs 'No' then output 'No $w \in \N$ such that $\CAL{A} \cong \CAL{M}_w$'. Otherwise, let $B_0, B_1, B_2, B_{3}$ be the output of the algorithm such that $f = \IMM_{w,4}(B_0\vecx_0, B_1\vecx_1, B_2\vecx_2, B_3\vecx_3)$.

~
\State Check if there exist matrices $F_{1,1}, \ldots ,F_{w,w} \in \CAL{M}_w$ such that $B_0\cdot L_{i,j}^T \cdot B_0^{-1} = I_w \otimes F_{i,j}^T$ and $B_1\cdot L_{i,j} \cdot B_1^{-1} = I_w \otimes F_{i,j}$ for all $i,j\in [w]$. If such matrices do not exist then output 'No $w \in \N$ such that $\CAL{A} \cong \CAL{M}_w$', otherwise output $\phi: \CAL{A} \rightarrow \CAL{M}_w$, where $\phi(E_{i,j}) = F_{i,j}$ for all $i,j \in [w]$ (extended linearly to the whole of $\CAL{A}$) as the algebra isomorphism from $\CAL{A}$ to $\CAL{M}_w$. 
\end{algorithmic}
\end{algorithm}
\vspace{0.1in}

\noindent \textbf{Steps 1--2}: At Step 2 there is a $K \in \GL(w^2, \F)$ and a basis $\{C_{1,1}, \ldots , C_{w,w}\}$ of $\CAL{M}_w$  such that $L_{i,j} = K^{-1}\cdot (I_w \otimes C_{i,j}) \cdot K$ for all $i,j \in [w]$ (by the Skolem-Noether theorem, see next claim).
\begin{claim}\label{claim: left multiplication algebra of the given algebra}
Suppose $\CAL{A} \cong \CAL{M}_w$ for some $w \in \N$. Then there exists a $K \in \GL(w^2, \F)$ and linearly independent matrices $\{C_{1,1}, \ldots , C_{w,w}\}$ in $\CAL{M}_w$  such that $L_{i,j} = K^{-1}\cdot (I_w \otimes C_{i,j}) \cdot K$~ for all $i,j \in [w]$.
\end{claim}
\noindent \textbf{Step 3}: The space spanned by $\{ L_{1,1}^T, \ldots , L_{w,w}^{T}\}$ is $K^T\cdot (I_{w} \otimes \CAL{M}_w)\cdot (K^T)^{-1}$. 
\begin{observation}\label{claim: commutator space of tensor product}
The space of matrices in $\CAL{M}_{w^2}$ that commute with every matrix in $K^T\cdot (I_{w} \otimes \CAL{M}_w)\cdot (K^T)^{-1}$ is $K^T\cdot (\CAL{M}_w \otimes I_w)\cdot (K^T)^{-1}$. So, $\{N_{1,1}, \ldots, N_{w,w}\}$ is a basis of $K^T\cdot (\CAL{M}_w \otimes I_w)\cdot (K^T)^{-1}$.
\end{observation}
\noindent \textbf{Step 4}: Let $n = 4w^2$. For $k\in [0,3]$, let $\CAL{B}'_k$ be the following spaces: Every matrix in $\CAL{B}'_k$ is a $n\times n$ block-diagonal matrix (with rows and columns indexed by $\vecx_0, \ldots, \vecx_3$) and its non-zero entries are confined to the rows and columns indexed by $\vecx_k$ and $\vecx_{k+1}$. For $B \in \CAL{B}_k$, let $[B]_k$ be the $2w^2 \times 2w^2$ sub-matrix of $B$ as defined in Equation \ref{equation: lie algebra matrices} (Section \ref{section: lie algebra of imm}). Then
\begin{eqnarray*}
\CAL{B}'_k &:=& \left\{ B \in \CAL{M}_n ~:~ [B]_k =  \begin{bmatrix} K^T\cdot (I_w \otimes M^T) (K^T)^{-1} & \mathbf{0} \\ \mathbf{0}  & K^{-1}\cdot (-I_w \otimes M) \cdot K \end{bmatrix}  ~\text{for}~ M \in \CAL{M}_w \right\} ~~\text{if $k$ is even}, \nonumber \\
&:=& \left\{ B \in \CAL{M}_n ~:~ [B]_k =  \begin{bmatrix} K^{-1} \cdot (M^T \otimes I_w) \cdot K & \mathbf{0} \\ \mathbf{0}  & K^{T}\cdot (-M \otimes I_w) \cdot (K^T)^{-1} \end{bmatrix}  ~\text{for}~ M \in \CAL{M}_w \right\} ~~\text{if $k$ is odd}.
\end{eqnarray*}

The following observation follows from Lemma \ref{theorem: lie algebra of of trace} and Fact \ref{fact: lie algebra conjugacy}.
\begin{observation}\label{observation: Lie algebra of the polynomial computed in the reduction}
The Lie algebra of $\IMM_{w,4}((K^T)^{-1}\vecx_0, K\vecx_1, (K^T)^{-1}\vecx_2, K\vecx_3)$ contains $\CAL{B}'_0, \CAL{B}'_1, \CAL{B}'_2, \CAL{B}'_3$. 
\end{observation}
At Step 4, Algorithm \ref{algorithm: reduction from fmai to trace equivalence} computes a non-zero $4$-tensor $f$ such that $\CAL{B}'_k \subseteq \G_f$ for all $k \in [0,3]$. Equation \ref{equation: for k odd} ensures $\CAL{B}'_0$, $\CAL{B}'_{2} \in \G_f$, and Equation \ref{equation: for k even} ensures $\CAL{B}'_1, \CAL{B}'_3 \in \G_f$. 
That the algorithm is able to compute a non-zero $f$ (by solving a linear system) follows from Observation \ref{observation: Lie algebra of the polynomial computed in the reduction}. Since the number of monomials in $f$ is at most $w^8$, this step runs in polynomial time.
 \vspace{0.1in}

\noindent \textbf{Step 5}: From Corollary \ref{corollary: lie algebra characterization for equivalent polynomials} it follows that $f(\vecx) = \alpha\cdot \IMM_{w,4}((K^T)^{-1}\vecx_0,$ $ K\vecx_1,$ $(K^T)^{-1}\vecx_2, K\vecx_3)$ for some $\alpha \in \F^{\times}$. Hence, at step 5 with high probability Algorithm \ref{algorithm: reduction from degree d to degree 3} outputs four matrices $B_0,B_1,B_2, B_3 \in \GL(w^2,\F)$ such that $f(\vecx) = \IMM_{w,4}$ $(B_0\vecx_0,$ $B_1\vecx_1, B_2\vecx_2, B_3\vecx_3)$.\vspace{0.1in}

\noindent \textbf{Step 6}: Let $B$ be the block-diagonal matrix whose $k$-th block is $B_k$, for $k \in [0,3]$. Since $\CAL{B}'_0 \subseteq \G_f$ and $\G_f = B^{-1}\cdot \GIMM \cdot B$ (from Fact \ref{fact: lie algebra conjugacy}), $B\cdot \CAL{B}'_0 \cdot B^{-1} \subseteq  \GIMM$. Observe that every matrix in $B\cdot \CAL{B}'_0 \cdot B^{-1}$ is block-diagonal with its non-zero entries confined to the first two blocks. Hence, from Lemma \ref{lemma:block-diagonal matrix with two non-zero blocks in GIMM is in B_k}, and the fact that both the spaces $B\cdot \CAL{B}'_0 \cdot B^{-1}$ and $\CAL{B}_0$ have dimension $w^2$, we have $B\cdot \CAL{B}'_0 \cdot B^{-1} = \CAL{B}_0$. In particular, for every $i,j \in [w]$ there is an $F_{i,j} \in \CAL{M}_w$ such that $B_0\cdot L^T_{i,j}\cdot B_{0}^{-1} = I_w\otimes F_{i,j}^T$ and $B_1\cdot L_{i,j}\cdot B_{1}^{-1} = I_w\otimes F_{i,j}$. Finally, verify that $\phi(E_{i,j}) = F_{i,j}$ is an algebra isomorphism.  \\

\emph{Comparison with \cite{GargGKS19}}: In \cite{GargGKS19}, $\FMAI$ is reduced to $\EDET$ by using the fact that $\Det_w$ is characterized by its Lie algebra (see Lemma 7.1 in \cite{GargGKS19}). If the input algebra $\CAL{A}$ is isomorphic to $\CAL{M}_w$ then the algorithm in \cite{GargGKS19} computes a \emph{degree-$w$} polynomial $f$ in $w^2$ variables such that $\G_f$ contains the Lie algebra of a polynomial equivalent to $\Det_w$. Hence, the time complexity of their algorithm is $w^{O(w)}$. Algorithm \ref{algorithm: reduction from fmai to trace equivalence} follows the same approach, but computes a \emph{degree four} polynomial $f$ such that $\G_f$ contains the Lie algebra of a polynomial equivalent to $\IMM_{w,4}$. So, the complexity of this algorithm is $w^{O(1)}$.

%% file: secappendix_ABPprelim.tex
\section{Preliminaries on algebraic branching programs and matrix products}{\label{subsection: abps and matrix products}}
\textbf{Set-multilinear polynomial}: A \emph{set-multilinear monomial} in $\vecx_0, \ldots, \vecx_{d-1}$ variables has exactly one variable from $\vecx_k$, for all $k \in [0,d-1]$. The coefficient of a non set-multilinear monomial is zero in a \emph{set-multilinear polynomial} in $\vecx_0, \ldots, \vecx_{d-1}$ variables. \vspace{0.1in}

The following definition is motivated from the the fact that monomials in $\IMM_{w,d}$ correspond to a path in the $d$-layer graph capturing the matrix product $Q_0\ldots Q_{d-1}$.
\begin{definition}[Path monomial]{\label{definition :path monomial}}
A set-multilinear monomial in $\vecx_0, \ldots, \vecx_{d-1}$ variables is called as a \emph{path monomial} if it has a non-zero coefficient in the $\IMM_{w,d}$ polynomial, and a set-multilinear monomial that is not a path monomial is called a \emph{non-path monomial}.
\end{definition}

\textbf{Linear matrices}: A matrix with entries as linear forms in $\vecx$ variables over $\F$ is called a linear matrix in $\vecx$ variables over $\F$. If $\vecx$, $\F$ are clear from the context, then it is simply called a linear matrix. If the linear forms in a linear matrix are linearly independent, then we say it is a \emph{full-rank} linear matrix. \vspace{0.1in}

\textbf{Algebraic branching program (ABP)}: A $(w,d,n)$-ABP is a matrix product $Y_0 \cdot Y_1 \ldots Y_{d-1}$, where $Y_0$ and $Y_{d-1}$ are row and column linear matrices of size $w$, and $Y_{k}$ is a $w\times w$ linear matrix in $\vecx$ variables for $k\in [1,d-2]$. The polynomial computed by the ABP is the entry in the resulting $1\times 1$ matrix. Note that in the general definition of an ABP the intermediate widths of matrices can vary, but throughout this article we work with uniform width ABPs unless stated otherwise. A \emph{full-rank} ABP is a $(w,d,n)$-ABP where the $w^2(d-2)+2w$ linear forms in its matrices are linearly independent. A \emph{set-multilinear} ABP in $\vecx_0, \ldots, \vecx_{d-1}$ variables is a $(w,d,n)$-ABP where the linear forms in $Y_k$ are in $\vecx_k$ variables. The following fact is easily inferred.
\begin{fact}\label{fact: ABP computing IMM}
The $\IMM_{w,d}$ polynomial is computed by a $(w^2,d,n)$-set-multilinear ABP $Y_0\ldots Y_{d-1}$ in $\vecx_0,\ldots ,$ $\vecx_{d-1}$ variables, where $Y_0 = (Q_0(1,1), Q_0(1,2), \ldots Q_0(1,w), Q_0(2,1), \ldots , Q_0(w,w))$, $Y_k = I_w\otimes Q_k$ for $k \in [1,d-2]$, and $Y_{d-1} = (Q_{d-1}(1,1), Q_{d-1}(2,1), \ldots Q_{d-1}(w,1), Q_{d-1}(1,2), \ldots , Q_{d-1}(w,w))^T$.
\end{fact}

\textbf{Matrix Product}: A matrix product $X_0\ldots X_{d-1}$, where $X_0,\ldots, X_{d-1}$ are $w\times w$ linear matrices is denoted as a $(w,d,n)$-matrix product. If the $w^2d$ linear forms in the matrices of a $(w,d,n)$-matrix product are linearly independent then we say it is a \emph{full-rank} $(w,d,n)$-matrix product. Additionally, if $X_k$ has linear forms in only $\vecx_k$ variables for $k\in [0,d-1]$ then we call it a $(w,d,n)$ set-multilinear matrix product in $\vecx_0, \ldots, \vecx_{d-1}$ variables. \vspace{0.1in}

\textbf{Set-multilinear ABP reconstruction}: Here, we note the main properties of the set-multilinear ABP reconstruction algorithm in \cite{KlivansS03} for set-multilinear ABPs with varying intermediate widths. A set-multilinear ABP $Y_0\ldots Y_{d-1}$ in $\vecx_0, \ldots , \vecx_{d-1}$ variables has width-sequence $(w_0, \ldots , w_{d-2})$, if $Y_0$ is a row linear matrix of size $w_0$ in $\vecx_0$ variables, $Y_k$ is a $w_{k-1}\times w_{k}$ linear matrix in $\vecx_k$ variables for $k\in [1,d-2]$, and $Y_{d-1}$ is a column linear matrix of size $w_{d-2}$ in $\vecx_{d-1}$ variables. The next observation is proved using evaluation dimension (see Definition \ref{definition: evaldim}), its proof is omitted here.
\begin{observation}\label{observation: min width ABP}
Suppose $f$ is a set-multilinear polynomial in $\vecx_0, \ldots ,\vecx_{d-1}$ variables. Then, there is a set-multilinear ABP in $\vecx_0, \ldots , \vecx_{d-1}$ variables of width sequence $(w_0, \ldots , w_{d-2})$ computing $f$, such that any other set-multilinear ABP in $\vecx_0, \ldots , \vecx_{d-1}$ variables of width-sequence $(w'_0, \ldots ,$ $w'_{d-2})$ computing $f$ satisfies $w_k \leq w'_k$ for $k\in [0,d-2]$. Such a set-multilinear ABP  in $\vecx_0, \ldots , \vecx_{d-1}$ variables of width-sequence $(w_0, \ldots , w_{d-2})$ computing $f$ is called a min-width set-multilinear ABP for $f$.
\end{observation}
Given blackbox access to a set-multilinear polynomial $f(\vecx_0, \ldots ,\vecx_{d-1})$, the set-multilinear ABP reconstruction algorithm in \cite{KlivansS03} reconstructs a min-width set-multilinear ABP in $\vecx_0, \ldots ,$ $ \vecx_{d-1}$ variables computing $f$ in randomized polynomial-time. Finally, the following observation regarding the relation between two min-width ABPs computing $f$ is easy to prove and its proof is omitted.
\begin{observation}\label{observation: min width ABP}
Suppose $f$ is a set-multilinear polynomial in $\vecx_0, \ldots ,\vecx_{d-1}$ variables, and $Y_0\ldots Y_{d-1}$ and $Y'_0\ldots Y'_{d-1}$ are two min-width set-multilinear ABPs in $\vecx_0, \ldots, \vecx_{d-1}$ variables of width-sequence $(w_0, \ldots$ $ ,w_{d-2})$ computing $f$. Then there are matrices $T_k  \in \GL(w_k,\F)$ $k\in [0,d-2]$, such that $Y'_0 = Y_0\cdot T_0$, $Y'_k = T^{-1}_{k-1}\cdot Y_k \cdot T_k$ for $k\in [1,d-2]$, and $Y'_{d-1} = T_{d-2}^{-1}\cdot Y_{d-1}$. 
\end{observation}

%% file: secappendix_lie_algebra_IMM.tex
\section{Proofs from Section \ref{section: lie algebra of imm}}\label{secappendix: lie algebra of imm}
\emph{Claim \ref{claim: gtrimm is block diagonal} (restated): If $E \in \G_{\IMM}$ then $E$ is block-diagonal.}
\begin{proof}
Since $E \in \GIMM$, the entries of $E = (e_{i,j})_{i,j\in [n]}$ satisfy the following equation, 
\begin{equation}{\label{equation:Lie Algebra equation for Tr-IMM}}
    \sum_{i,j \in [n]}e_{ij}\cdot x_{j}\cdot\frac{\partial \IMM}{\partial x_i} = 0\ \ .
\end{equation}
Equation \ref{equation:Lie Algebra equation for Tr-IMM}  can be rewritten as follows
\begin{equation}{\label{equation:main eqn split}}
\underbrace{\sum_{\substack{x_i,x_j \in \vecx_k \\ k\in[0,d-1]} }e_{ij} \cdot x_j \cdot \frac{\partial \IMM}{\partial x_i}}_{(a)}~~ + ~~~\underbrace{\sum_{\substack{x_i \in \vecx_{\ell}\ x_j \in \vecx_k \\ \ell,k\in [0,d-1], \ell \neq k}} e_{ij} \cdot x_j \cdot \frac{\partial  \IMM}{\partial x_i}}_{(b)}  = 0 .
\end{equation} 
In Equation \ref{equation:main eqn split}, term (a) corresponds to the \textit{block-diagonal} entries of $E$ and term (b) corresponds to the \textit{non block-diagonal} entries of $E$. Observe that the terms are monomial disjoint: monomials in term (a) have variables from each variable set $\vecx_{0}, \ldots, \vecx_{d-1}$, whereas monomials in term (b) have two variables from $\vecx_{k}$ and no variable from $\vecx_{\ell}$ for $\ell,k\in [0,d-1]$ and $\ell\neq k$. This implies terms (a) and (b) are individually equal to zero, 
\begin{equation}{\label{equation: split term a are zero}}
\sum_{\substack{x_i,x_j \in \vecx_k \\ k\in[0,d-1]} }e_{ij} \cdot x_j \cdot \frac{\partial \IMM}{\partial x_i} = 0
\end{equation}
\begin{equation}{\label{equation: split term b are zero}}
\sum_{\substack{x_i \in \vecx_{\ell}\ x_j \in \vecx_k \\ \ell,k\in [0,d-1], \ell \neq k}} e_{ij} \cdot x_j \cdot \frac{\partial  \IMM}{\partial x_i} = 0.
\end{equation}
Additionally in Equation \ref{equation: split term b are zero}, for $x_i\in \vecx_{\ell}, x_j \in \vecx_{k}$, $x_{i'}\in \vecx_{\ell'}, x_{j'} \in \vecx_{k'}$ and $(\ell,k) \neq (\ell',k')$ the terms $x_{j} \frac{\partial \IMM}{\partial x_i}$ and $x_{j'} \frac{\partial \IMM}{\partial x_{i'}}$ are monomial disjoint. Thus for every pair $(\ell,k)$ such that $\ell\neq k$
\begin{align}{\label{equation: split terms for non-block diagonal}}
\sum_{\substack{x_i \in \vecx_{\ell}\ x_j \in \vecx_k }} e_{ij} \cdot x_j \cdot \frac{\partial  \IMM}{\partial x_i} &= 0.
\end{align}
In Equation \ref{equation: split terms for non-block diagonal}, group the coefficients of the term $\frac{\partial \IMM}{\partial x_i}$ together and rewrite it as
\begin{equation}{\label{equation: non block diagonal space linear form}}
     \sum_{x_i \in \vecx_{\ell}} L_{x_i}^{(\ell,k)}\frac{\partial \IMM}{\partial x_i} = 0,
\end{equation}
where $L_{x_i}^{(\ell,k)}$ is a linear form in the $\vecx_k$ variables. Now we show that $L_{x_i}^{(\ell,k)} = 0$. \vspace{0.1in}

Let $x_i = x_{p,q}^{(\ell)}$ be the $(p,q)$-th entry of $Q_{\ell}$, where $p,q \in [w]$. Also let $Q'_{\ell}$ be a $w \times w$ matrix whose $(p,q)$-th entry is the linear form $L_{x_i}^{(\ell,k)}$. Then from Equation \ref{equation: non block diagonal space linear form},
\begin{equation}\label{equation: non block diagonal space linear form with trace}
    \sum_{x_i \in \vecx_{\ell}}L_{x_i}^{(\ell,k)}\frac{\partial \IMM}{\partial x_i}~ =~ \text{tr}(Q_0 \ldots Q'_{\ell} \ldots Q_{d-1}) ~=~ 0~.
\end{equation} 
Now suppose for contradiction $L_{x_i}^{(\ell,k)}\neq 0$. Then there is a $x_{u,v}^{(k)} \in \vecx_k$ such that the coefficient of $x_{u,v}^{(k)}$ in $L_{x_i}^{(\ell,k)}$ is non-zero. We argue for the cases $k \notin \{\ell-1, \ell+1\}$ and $k \in \{\ell-1, \ell+1\}$ separately. If $k \notin \{\ell-1, \ell+1\}$ then a path monomial $\mu$ can be chosen such that $\mu$ contains the variable $x_i = x_{p,q}^{(\ell)}$ and $x_{u,v}^{(k)}$. In Equation \ref{equation: non block diagonal space linear form with trace} set all the variables to zero except the variables appearing in $\mu$. Under this assignment the polynomial computed by $\text{tr}(Q_0 \ldots Q'_{\ell} \ldots Q_{d-1})$ is non-zero as the linear form $L_{x_i}^{(\ell,k)} \neq 0$, which is a contradiction. Now suppose $k = \ell - 1$. Then choose a path monomial $\mu$ containing the variables $x_i = x_{p,q}^{(\ell)}$ and $x_{u',p}^{(k)}$ where $u\neq u'$, and in Equation \ref{equation: non block diagonal space linear form with trace} set all the variables to zero except the variables appearing in $\mu$ and the variable $x_{u,v}^{(k)}$. Again under this assignment the polynomial computed by $\text{tr}(Q_0 \ldots Q'_{\ell} \ldots Q_{d-1})$ is non-zero as the linear form $L_{x_i}^{(\ell,k)} \neq 0$, which is a contradiction. For $k = \ell+1$, choosing a path monomial $\mu$ containing the variables $x_i = x_{p,q}^{(\ell)}$ and $x_{q,v'}^{(k)}$ where $v\neq v'$ suffices.
\end{proof}
\emph{Lemma \ref{theorem: lie algebra of of trace} (restated): The space $\CAL{B}_0 + \ldots + \CAL{B}_{d-1}$ is contained in $\GIMM$.}
\begin{proof}
It is sufficient to prove that for every $k \in [0,d-1]$, $\CAL{B}_k \subseteq \G_{\IMM}$. Let $k \in [0,d-2]$, $k$ even, and $B \in \CAL{B}_k$. Then there is an $M \in \CAL{M}_w$ such that 
$$ [B]_k = \begin{bmatrix}
     I_w \otimes M^T & \mathbf{0} \\
     \mathbf{0}  & -I_w \otimes M
    \end{bmatrix}~~.
$$
Let $M = (m_{i,j})_{i,j \in [w]}$, and $\ell_{i,j}^{(k)} = \sum_{v\in [w]} m_{v,j} x_{i,v}^{(k)}$ and $\ell_{i,j}^{(k+1)} = \sum_{v\in [w]} - m_{i,v} x_{v,j}^{(k+1)}$ for all $i,j \in [w]$. Further, let $Q'_k = (\ell_{i,j}^{(k)})_{i,j \in [k]}$, and $Q'_{k+1} = (\ell_{i,j}^{(k+1)})_{i,j\in [w]}$. 
\begin{observation}\label{observation: containment of B}
The matrix $B \in \GIMM$ if and only if the following holds:
\begin{align*}\label{equation: lie algebra equation for block spaces in Lie algebra lemma }
    \sum_{i,j\in [w]} \ell_{i,j}^{(k)} \frac{\partial \IMM}{x_{i,j}^{(k)}} ~~~+ \sum_{i,j\in [w]} \ell_{i,j}^{(k+1)} \frac{\partial \IMM}{x_{i,j}^{(k+1)}} ~~~ &= ~~~  \tr(Q_0\cdots Q_{k-1}( Q'_k\cdot Q_{k+1} +  Q_k\cdot Q'_{k+1})Q_{k+2} \cdots Q_{d-1}) \\
    & = ~~~~0~.
\end{align*}
\end{observation}
\begin{observation}\label{observation: structure of Q'k}
The matrices $Q'_k$ and $Q'_{k+1}$ are such that $Q'_k = Q_k\cdot M$ and $Q'_{k+1} = -M\cdot Q_{k+1}$.
\end{observation}
Thus, $Q'_k\cdot Q_{k+1} = -Q_{k}\cdot Q'_{k+1}$, and hence $B \in \G_{\IMM}$. The proofs for the remaining two cases: a) $k \in [0,d-1]$, $d$ even and $k$ odd, and b) $k = d-1$ and $d$ odd are similar.
\end{proof}
\emph{Lemma \ref{lemma:block-diagonal matrix with two non-zero blocks in GIMM is in B_k} (restated): Suppose $B\in \GIMM$ and there is a $k\in [0,d-1]$ such that the non-zero entries of $B$ are confined to the rows and columns that are indexed by $\vecx_k$ and $\vecx_{k+1}$ variables. Then  $B\in \CAL{B}_k$.}
\begin{proof}
Let $\ell_{i,j}^{(k)}$ and $\ell_{i,j}^{(k+1)}$ be the linear forms whose coefficients are given by the row vectors indexed by $x_{i,j}^{(k)}$  and $x_{i,j}^{(k+1)}$ variables in $B$ respectively. From the structure of $B$ it follows that $\vecx_k$ and $\vecx_{k+1}$ are the only variables with non-zero coefficients in $\ell_{i,j}^{(k)}$ and $\ell_{i,j}^{(k+1)}$ respectively. Let $Q'_k = (\ell_{i,j}^{(k)})_{i,j\in [w]}$, and $Q'_{k+1} = (\ell_{i,j}^{(k+1)})_{i,j\in [w]}$. Since $B\in \GIMM$,
\begin{equation}\label{equation: lie algebra equation for B}
\text{tr}(Q_0\ldots Q'_k\cdot Q_{k+1}\ldots Q_{d-1} + Q_0\ldots Q_k\cdot Q'_{k+1}\ldots Q_{d-1}) = 0 ~~.
\end{equation}
\begin{observation}\label{observation: Q'kQk=QkQ'k}
Equation \ref{equation: lie algebra equation for B} implies  $Q'_k\cdot Q_{k+1} + Q_k\cdot Q'_{k+1} = 0$
\end{observation}
\begin{proof}
The third line in the following sequence of equations follows from the fact that trace remains invariant under rotations.
\begin{align*}
 \text{tr}(Q_0\ldots Q'_k\cdot Q_{k+1}\ldots Q_{d-1} + Q_0\ldots Q_k\cdot Q'_{k+1}\ldots Q_{d-1})\\
  = \text{tr}(Q_0\ldots Q_{k-1}(Q'_k\cdot Q_{k+1} +  Q_k\cdot Q'_{k+1}) Q_{k+2}\ldots Q_{d-1}) \\
 = \text{tr}((Q'_k\cdot Q_{k+1} +  Q_k\cdot Q'_{k+1})Q_{k+2}\ldots Q_{d-1}\cdot Q_0 \ldots Q_{k-1}) = 0 ~.
\end{align*}
Assign $0/1$ values to the variables in $\vecx_{k+3}, \ldots, \vecx_{d-1}, \vecx_{0}, \ldots, \vecx_{k-1}$ such that $Q_{k+3}, \ldots, Q_{d-1}, Q_{0}, \ldots, Q_{k-1}$ become identity matrices under this assignment. As the entries of $Q_{k+2}$ are distinct $\vecx_{k+2}$ variables, we have $Q'_k\cdot Q_{k+1} +  Q_k\cdot Q'_{k+1} = 0$.
\end{proof}
By setting the $\vecx_{k+1}$ variables to $0/1$ so that $Q_{k+1}$ becomes identity in the equation $Q'_k\cdot Q_{k+1} +  Q_k\cdot Q'_{k+1} = 0$, we get $Q'_{k} = Q_1 M$ for some $M \in \CAL{M}_w$. Similarly, $Q'_{k+1} = N Q_{k+1}$. Thus, $Q_kM Q_{k+1} +  Q_kN Q_{k+1} = 0$ which implies $N = -M$. At this point, the structure of $B$ can be determined using $M$ and then it is easily observed that $B\in \CAL{B}_k$.
\end{proof}
\begin{claim}{\label{claim:diag_distinct}}
There is a diagonal matrix in $\GIMM$ with distinct diagonal entries.
\end{claim}
\begin{proof}
For $k\in [0,d-1]$, let $D_k$ be a $w\times w$ diagonal matrix whose $i$-th diagonal entry is denoted as $d_i^{(k)}$. For $k\in [0,d-1]$ let $B_k \in \CAL{B}_k$ be the diagonal matrix whose $2w^2\times 2w^2$ sub-matrix indexed by $\vecx_{k}\uplus \vecx_{k+1}$ variables, denoted $[B_k]_k$, looks as follows: if $d$ is even then
\begin{eqnarray} 
[B_k]_k &=&  \begin{bmatrix} I_w \otimes D_k & \mathbf{0} \\ \mathbf{0}  & -I_w \otimes D_k \end{bmatrix}   ~~\text{if $k$ is even}, \nonumber \\
&=&  \begin{bmatrix} D_k \otimes I_w & \mathbf{0} \\ \mathbf{0}  & -D_k \otimes I_w \end{bmatrix}  ~~\text{if $k$ is odd}.
\end{eqnarray}
If $d$ is odd, then $B_0,\ldots , B_{d-2}$ remain the same, and only $B_{d-1}$ is defined differently and in this case
\begin{equation*}
[B_{d-1}]_{d-1} = \begin{bmatrix} I_w \otimes D_{d-1} & \mathbf{0} \\ \mathbf{0}  & -D_{d-1} \otimes I_w \end{bmatrix}  ~.
\end{equation*}
Suppose $B = \sum_{k=0}^{d-1}B_k$. Then $B$ is a diagonal matrix in $\GIMM$ (from Lemma \ref{theorem: lie algebra of of trace}) whose diagonal entry indexed by the variable $x_{i,j}^{(k)}$ is equal to $d_{j}^{(k)} - d_{i}^{(k-1)}$. If we pretend the entries of $D_0, \ldots, D_{d-1}$ to be formal variables, say $\mathbf{d}$ variables, then the $n$ diagonal entries of $B$ are $n$ distinct linear forms in $\mathbf{d}$ variables. Hence, if we assign values to the $\mathbf{d}$ variables uniformly at random from a set $S \subseteq \F$  such that $|S|\geq n^3$ then with non-zero probability $B$ has all diagonal entries distinct after the random assignment. 
\end{proof}
\begin{lemma}{\label{lemma:char_poly_sq_free}}
Let $E_1, \ldots, E_a$ be a basis of $\GIMM$ and $E = \sum_{i = 1}^{a} r_iE_i$ where $r_i \in_r S \subset \mathbb{F}, |S| \geq n^4$. Then the characteristic polynomial of $E$ is square-free with probability $1 - o(1)$.
\end{lemma}
\begin{proof}
If we treat $\vecr = \{r_1, \ldots, r_a\}$ as formal variables then the characteristic polynomial $h_r(x)$ of $E$ is a polynomial in $x$ with coefficients that are polynomial of degree at most $n$ in $\vecr$ variables. Observe that the discriminant of $h_r(x)$, denoted $\text{disc}(h_r(x))$ is a non-zero polynomial in the $\vecr$ variables of degree at most $2n^2$. This is because if $\text{disc}(h_r(x))$ is identically zero as polynomial in $\vecr$ variables then for every evaluation of $\vecr$ variables to field elements, $h_r(x)$ is not a square-free polynomial. This contradicts Claim \ref{claim:diag_distinct}, as we can set the $\vecr$ variables appropriately such that $E$ is a diagonal matrix with distinct entries and $h_r(x)$ for such a setting is square-free. Since $\text{disc}(h_r(x))$ is not an identically zero polynomial in $\vecr$ variables and has degree less than $2n^2$, if we set the $\vecr$ variables independently and uniformly at random from $S\subseteq \F$, $|S| \geq 2n^3$ then with probability $1-o(1)$ $\text{disc}(h_r(x)) \neq 0$, i.e., $h_r(x)$ is square-free.
\end{proof}

\emph{Claim \ref{claim:invar_subspace_is_coordinate} (restated): Any non-zero $\GIMM$-invariant subspace is a coordinate subspace of $\F^n$.}
\begin{proof}
Let $\mathbf{u} = (u_1, \ldots, u_n) \in \mathcal{U}$, $S_\mathbf{u}$ be the set of non-zero coordinates of $\mathbf{u}$, that is $S_{\mathbf{u}} := \{ j : u_j \neq 0\  \textnormal{and}\  j \in [n] \}$, and $D$ be a diagonal matrix as in Claim \ref{claim:diag_distinct} with distinct diagonal entries $\lambda_1, \ldots, \lambda_n$. Then the vectors $\{(\lambda_1^iu_1, \ldots , \lambda_n^iu_n) \in \mathcal{U} \mid i\in [0,|S_{\vecu}|-1]\}$ are $\F$-linearly independent. Hence for all $j \in S_{\vecu}$, $e_j \in \mathcal{U}$ .
This implies $\CAL{U}$ is a coordinate subspace. 
\end{proof}
\emph{Lemma \ref{lemma:irreducible_space_limm} (restated): The only irreducible $\GIMM$-invariant subspaces of $\F^n$ are $\mathcal{U}_0, \ldots, \mathcal{U}_{d-1}$.}
\begin{proof} 
It follows from Claim \ref{claim: gtrimm is block diagonal} that $\mathcal{U}_0, \ldots, \mathcal{U}_{d-1}$ are $\GIMM$-invariant subspaces. We show that $\CAL{U}_k$ is irreducible for $k\in [0,d-1]$. Suppose $\mathcal{U}$ is a non-zero $\GIMM$-invariant subspace and $\CAL{U} \subset \mathcal{U}_k$ for some $k \in [0,d-1]$. Then $\CAL{U}$ is a coordinate subspace of $\F^n$ (from Claim \ref{claim:invar_subspace_is_coordinate}). Let $e_x \in \F^n$ be the coordinate vector with one in the entry indexed by the variable $x \in \vecx$. Then to prove that $\CAL{U} = \CAL{U}_k$, it is sufficient to show that $e_x \in \CAL{U}$ for all $x \in \vecx_k$. We show this when $d$ is even. (For $d$ odd the matrices $B_k$ and $B_{k-1}$ defined below need to be appropriately redefined for $k=d-1$  so that $B_{d-1} \in \CAL{B}_{d-1}$ and $B_{d-2} \in \CAL{B}_{d-2}$.) Let $1_w$ be the all ones $w \times w$ matrix. Define the matrices $B_k\in \mathcal{B}_k$ and $B_{k-1}\in \mathcal{B}_{k-1}$ as follows: If $k$ is odd then 
$$[B_k]_k = \begin{bmatrix}
    1_w \otimes I_w & 0 \\ 
    0 & -1_w \otimes I_w
  \end{bmatrix} ~~~\text{and}~~~
[B_{k-1}]_{k-1} = \begin{bmatrix}
    -I_w \otimes 1_w & 0 \\ 
    0 & I_w \otimes 1_w
  \end{bmatrix}.
$$
If $k$ is even then 
$$[B_k]_k = \begin{bmatrix}
    -I_w \otimes 1_w & 0 \\ 
    0 & I_w \otimes 1_w
  \end{bmatrix} ~~~\text{and}~~~
[B_{k-1}]_{k-1} = \begin{bmatrix}
    1_w \otimes I_w & 0 \\ 
    0 & -1_w \otimes I_w
  \end{bmatrix} .
$$
Consider the matrix $E = B_{k-1} + B_k$ in $\GIMM$. Since $\CAL{U}$ is a coordinate subspace, there is a $y = x_{i,j}^{(k)} \in \vecx_k$ such that $e_y \in \CAL{U}$. Observation \ref{obs:col_row_var_in_u} follows from the structure of $E$ and Claim \ref{claim:invar_subspace_is_coordinate}.
\begin{observation}{\label{obs:col_row_var_in_u}}
The entries of the vector $Ee_y$ indexed by the variables in $\{ x_{i,1}^{(k)}, x_{i,2}^{(k)}, \ldots, x_{i,w}^{(k)}\}$ and $\{ x_{1,j}^{(k)}, x_{2,j}^{(k)}, \ldots, x_{w,j}^{(k)} \}$ are one and hence the coordinate vectors corresponding to these variables are in $\CAL{U}$.
\end{observation}
Applying Observation~\ref{obs:col_row_var_in_u} repeatedly we have that $e_x \in \mathcal{U}$ for all $x \in \mathbf{x}_k$. Hence, $\mathcal{U} = \mathcal{U}_k$ implying $\CAL{U}_k$ is irreducible. Finally, we argue that $\CAL{U}_0, \ldots, \CAL{U}_{d-1}$ are the only irreducible $\GIMM$-invariant subspaces. Let $\mathcal{U}$ be an irreducible $\GIMM$-invariant subspace and hence a coordinate subspace of $\F^n$. Suppose $e_y \in \mathcal{U}$, where $y \in \mathbf{x}_k$ for some $k \in [0,d-1]$. Applying Observation~\ref{obs:col_row_var_in_u} repeatedly we have that $e_x \in \mathcal{U}$ for all $x \in \mathbf{x}_k$. Hence, $\mathcal{U}_k \subseteq \CAL{U}$. Since $\mathcal{U}$ is irreducible, $\CAL{U} = \CAL{U}_k$.
\end{proof}
\emph{Corollary \ref{corollary:irr_inv_sub_f} (restated): If $f = \IMM_{w,d}(A\vecx)$, where $A \in \GL(n,\F)$, then the only irreducible $\G_f$-invariant subspaces of $\F^n$ are $A^{-1}\CAL{U}_0, \ldots, A^{-1}\CAL{U}_{d-1}$. }
\begin{proof}
This follows by observing that $\CAL{U}$ is an irreducible $\GIMM$-invariant subspace if and only if  $A^{-1}\CAL{U}$ is an irreducible $\G_f$-invariant subspace 
Since $\CAL{U}_0, \ldots , \CAL{U}_{d-1}$ are the only irreducible $\GIMM$-invariant subspaces, $A^{-1}\CAL{U}_0,$ $\ldots, A^{-1}\CAL{U}_{d-1}$ are the only irreducible $\G_f$-invariant subspaces.
\end{proof}
\emph{Lemma \ref{fact: symmetries of IMM} (restated): Let $\IMM_{w,d} = \tr(Q'_0\cdots Q'_{d-1})$, where $Q'_0\cdots Q'_{d-1} $ is a full-rank $(w,d,n)$-matrix product in $\vecx$ variables over $\F$. Then there are $C_0, \ldots, C_{d-1} \in \GL(w,\F)$ and $\ell\in [0,d-1]$ such that either $Q'_k = C_k\cdot Q_{\ell+k} \cdot C_{k+1}^{-1}$ for $k \in [0,d-1]$ ~or~ $Q'_k = C_k\cdot Q_{\ell -k}^T \cdot C_{k+1}^{-1}$ for $k \in [0,d-1]$.}
\begin{proof}
The proof of Lemma \ref{fact: symmetries of IMM} uses the following observation, which is on the evaluation dimension (Definition \ref{definition: evaldim}) of a polynomial expressed as the trace of a full-rank set-multilinear matrix product. The proof of Observation \ref{observation: evaldim of trace of full-rank set-mult matrix product} is similar to the proof of Observation \ref{obs:sigma_compute}.
\begin{observation}\label{observation: evaldim of trace of full-rank set-mult matrix product}
Let $f =\text{tr}(X_0\ldots X_{d-1})$, where $X_0\ldots X_{d-1}$ is a full-rank $(w,d,n)$ set-multilinear matrix product in $\vecx_0, \ldots, \vecx_{d-1}$ variables. Then a) for $k\in [0,d-1]$ and $k'\in \{k-1,k+1\}$  $\textnormal{Evaldim}_{\vecx_{k} \uplus \vecx_{k'}}(f) = w^2$, and b) for $k\in [0,d-1]$ and $k'\in [0,d-1]\setminus \{k-1, k, k+1\}$  $\textnormal{Evaldim}_{\vecx_{k} \uplus \vecx_{k'}}(f) = w^4$.
\end{observation} 
Let $A\in \GL(n,\F)$ be such that the row of $A$ indexed by the $x_{i,j}^{(k)}$ variable determine the coefficients of the linear form in the $(i,j)$-th entry of $Q'_k$ for $i,j\in [w]$ and $k\in [0,d-1]$. Then $\IMM_{w,d} = \IMM_{w,d}(A\vecx)$. Observation \ref{obs: existence of sigma for A} proves that $A$ is a block-diagonal matrix up to a rotation. 
\begin{observation}\label{obs: existence of sigma for A}
There is a permutation $\sigma$ of $[0,d-1]$ such that the non-zero entries of the rows of $A$ indexed by the $\vecx_{k}$ variables are confined to the columns of $A$ indexed by $\vecx_{\sigma(k)}$ variables. Further, there is an $\ell \in [0,d-1]$ such that either $\sigma(k) = \ell+k$ for $k\in [0,d-1]$ or $\sigma(k) = \ell-k$ for $k\in [0,d-1]$.
\end{observation}
\begin{proof}
By Lemma \ref{corollary:irr_inv_sub_f}, the irreducible invariant subspaces of the Lie algebra of $\IMM_{w,d}(A\vecx)$ are $A^{-1}\CAL{U}_0, \ldots, A^{-1}\CAL{U}_{d-1}$. But the irreducible $\GIMM$-invariant subspaces are $\CAL{U}_0, \ldots, \CAL{U}_{d-1}$ (Lemma \ref{lemma:irreducible_space_limm}). Hence, there is a permutation $\sigma$ of $[0,d-1]$ such that $A^{-1}\CAL{U}_{k} = \CAL{U}_{\sigma(k)}$ for $k\in [0,d-1]$. Since $\CAL{U}_k$ is the subspace spanned by the vectors whose non-zero entries are indexed by $\vecx_k$ variables, the non-zero entries of the columns of $A^{-1}$ indexed by the $\vecx_{k}$ variables are confined to the rows of $A^{-1}$ indexed by $\vecx_{\sigma(k)}$ variables. Consequently, the non-zero entries of the rows of $A$ indexed by the $\vecx_{k}$ variables are confined to the columns of $A$ indexed by $\vecx_{\sigma(k)}$ variables. Hence, $Q'_0\ldots Q'_{d-1}$ is a full-rank $(w,d,n)$ set-multilinear matrix product in $\vecx_{\sigma(0)}, \ldots, \vecx_{\sigma(d-1)}$ variables.

For $k, k' \in [0,d-1]$, if $k' \in \{k-1, k+1\}$ then $\textnormal{Evaldim}_{\vecx_{k} \uplus \vecx_{k'}}(\IMM_{w,d}) = w^2$, and if $k'\in [0,d-1]\setminus \{k-1,k,k+1\}$ then $\textnormal{Evaldim}_{\vecx_{k} \uplus \vecx_{k'}}(f) = w^4$ (from Observation \ref{observation: evaldim of trace of full-rank set-mult matrix product}). Let $\sigma(0) = \ell$. Then again using Observation \ref{observation: evaldim of trace of full-rank set-mult matrix product} either $\sigma(k) = \ell+ k$ for $k\in [0,d-1]$, or $\sigma(k) = \ell-k$ for $k\in [0,d-1]$.
\end{proof}
Let $\sigma$ be as in Observation \ref{obs: existence of sigma for A}. We assume that there is an $\ell\in [0,d-1]$ such that $\sigma(k) = \ell+k$ for $k\in [0,d-1]$ and prove that there are matrices $C_0, \ldots, C_{d-1}, D_0, \ldots, D_{d-1} \in \GL(w,\F)$ and non-zero $\alpha_0, \ldots , \alpha_{d-1} \in \F $ such that $Q'_k = C_k\cdot Q_{\ell+k} \cdot D_k$ for $k \in [0,d-1]$, $D_{d-1}C_0 = \alpha_0I_w$, $D_{k}\cdot C_{k+1} = \alpha_{k+1}I_w$ for $k\in [0,d-2]$, and $\prod_{k\in [0,d-1]} \alpha_k =1$. Using a similar argument it can be shown that if $\sigma(k) = \ell-k$ for $k\in [0,d-1]$ then there are matrices $C_0, \ldots, C_{d-1}, D_0, \ldots, D_{d-1} \in \GL(w,\F)$ and non-zero $\alpha_0, \ldots , \alpha_{d-1} \in \F $ such that $Q'_k = C_k\cdot Q_{\ell-k}^T \cdot D_k$ for $k \in [0,d-1]$, $D_{d-1}C_0 = \alpha_0I_w$, $D_{k}\cdot C_{k+1} = \alpha_{k+1}I_w$ for $k\in [0,d-2]$, and $\prod_{k\in [0,d-1]} \alpha_k =1$. For ease of exposition, we also assume that $\ell=0$, and it can be easily verified that the arguments continue to hold for an arbitrary $\ell$. Notice that if $\ell=0$ then $A$ is a block-diagonal matrix. Denote the block of $A$ indexed by $\vecx_k$ variables as $A_k$. The proof of the lemma is now almost complete using Observation \ref{obs: the structure of the subblocks of A}.
\begin{observation}\label{obs: the structure of the subblocks of A}
For $k\in [0,d-1]$, there are matrices $P_k,S_k \in \GL(w,\F)$ such that $A_k = (I_w\otimes P_k)(S_k\otimes I_w)$.
\end{observation}
\begin{proof}
Fix a $k\in [0,d-1]$ such that $k$ is even. We will show that there are matrices $P_k,S_k \in \GL(w,\F)$ such that $A_k = (I_w\otimes P_k)(S_k\otimes I_w)$, and a similar argument shows that there are matrices  $P_{k+1},S_{k+1} \in \GL(w,\F)$ such that $A_{k+1} = (I_w\otimes P_{k+1})(S_{k+1}\otimes I_w)$. Since $A$ is block-diagonal, $A^{-1}\CAL{B}_kA = \CAL{B}_k$ for all $k\in [0,d-1]$, from Lemma \ref{lemma:block-diagonal matrix with two non-zero blocks in GIMM is in B_k}, and Fact \ref{fact: lie algebra conjugacy}. Hence, for every $M\in \CAL{M}_w$ there is a unique $N\in \CAL{M}_w$ such that 
$$(I_w\otimes M) A_k = A_k ( I_w\otimes N)~.$$ 
Call the $w\times w$ sub-matrix of $A_k$ whose rows are indexed by $x_{i,1}^{(k)}, \ldots , x_{i,w}^{(k)},$ variables, and the columns are indexed by  $x_{j,1}^{(k)},, \ldots , x_{j,w}^{(k)},$ variables as $A_k(i,j)$. Note that for all $i,j\in [w]$ the following holds: $M\cdot A_k(i,j)= A_{k}(i,j) \cdot N$. Since this holds for any $M \in \CAL{M}_w$, either $A_k(i,j)$ is invertible or $A_k(i,j)$ is the zero matrix for $i,j \in [w]$. Choose an $i,j \in [w]$ such that $A_k(i,j)$ is invertible, and let $P_k = A_k(i,j)$. Since $A_k$ is invertible there exists such an $i,j\in [w]$. Let $u,v\in [w]$ be such that $A_k(u,v)$ is invertible. Then for any $M\in \CAL{M}_w$, 
$$P_k^{-1}\cdot M \cdot P_k = A_k(u,v)^{-1} \cdot M \cdot A_k(u,v)~.$$
Since the above holds for any $M\in \CAL{M}_w$, there is a non-zero $s_{u,v} \in \F$ such that $A_{k}(u,v) = s_{u,v}P_k$. Let $S_k = (s_{u,v})_{u,v\in [w]}$, where $s_{i,j}=1$, and for any $u,v\in [w]$ if $A_{k}(u,v)$ is zero then $s_{u,v}=0$. It is easily observed that $A_k = (I_w\otimes P_k)(S_k\otimes I_w)$, and as $A_k$ is invertible $S_k$ is invertible.
\end{proof}
From Observation \ref{obs: the structure of the subblocks of A}, for $k\in [0,d-1]$ the following is true: if $k$ is even then $Q'_k = S_k\cdot Q_k \cdot P_k^T$, and if $k$ is odd then $Q'_k = P_k\cdot Q_k \cdot S_k^T$. For ease of notation, if $k$ is even then rename $P_k^T$ as $D_k$ and $S_k$ as $C_k$, and  if $k$ is odd then rename $P_k$ as $C_k$ and $S_k^T$ as $D_k$. Hence for $k\in [0,d-1]$, $Q'_k = C_k\cdot Q_k\cdot D_k$. Now, observe the following
\begin{align*}
\text{tr}(Q_0\ldots Q_{d-1}) &= \text{tr}(C_0\cdot Q_0\cdot D_0 \ldots C_{d-1}\cdot Q_{d-1}\cdot D_{d-1} ) \\
&=  \text{tr}(D_{d-1}\cdot C_0\cdot Q_0\cdot D_0 \ldots C_{d-1}\cdot Q_{d-1})
\end{align*}
The last line line in the above equation follows from the fact that trace of a matrix product remains invariant under rotations. Since the entries of $Q_{d-1}$ are distinct variables disjoint from the variables in $Q_0, \ldots ,Q_{d-2}$
$$Q_0\ldots Q_{d-2} = D_{d-1}\cdot C_0\cdot Q_0\cdot D_0 \ldots Q_{d-2}\cdot D_{d-2}\cdot C_{d-1}~.$$
Substitute $Q_{k} = (D_{k}\cdot C_{k+1})^{-1}$ for $k\in [2,d-2]$, and $Q_1 = (D_0\cdot C_1)^{-1}(D_{1}\cdot C_{2})^{-1}$ in the above equation, and let $M = \prod_{k\in [0,d-2]} (D_{k}\cdot C_{k+1})^{-1}$. Then
$$Q_0 \cdot M =  D_{d-1}\cdot C_0\cdot Q_0~.$$
Since the entries of $Q_0$ are distinct variables, there is a non-zero $\alpha_0 \in \F$ such that $D_{d-1}\cdot C_0 = M =  \alpha_0 I_w$. Similarly, it can be shown that there is a non-zero $\alpha_{k+1} \in \F$ such that $D_k\cdot C_{k+1} = \alpha_{k+1} I_w$ for $k\in [0,d-2]$. Moreover, as 
$$\text{tr}(Q_0\ldots Q_{d-1}) = \text{tr}(C_0\cdot Q_0\cdot D_0 \ldots C_{d-1}\cdot Q_{d-1}\cdot D_{d-1})$$ 
it follows that $\prod_{k\in [0,d-1]} \alpha_k = 1$. Finally, observe the following
$$ Q'_{k} = \left( (\prod_{\ell \in [k+1,d-1]} \alpha_{\ell}) C_k\right) \cdot Q_k \cdot  \left( (\prod_{\ell \in [k+1,d-1]} \alpha^{-1}_{\ell}) D_k\right) ~~~ \text{for} ~~ k\in [0,d-2] .$$
Reusing symbols for ease of notation, rename $C_k$ as $(\prod_{\ell \in [k+1,d-1]} \alpha_{\ell}) C_k$, and $D_k$ as $(\prod_{\ell \in [k+1,d-1]} \alpha^{-1}_{\ell}) D_k$, and notice that $D_k = C_{k+1}^{-1}$ for $k\in [0,d-2]$, and $D_{d-1} = C_0^{-1}$.
\end{proof}

%% file: secappendix_reduction_from_trace_to_mult_equivalence.tex
\section{Reduction from $\EIMM$ to $\IMMTI$}{\label{secappendix: reduction to multilinear equivalence testing}}
\begin{algorithm}
\caption{Reduction from $\EIMM$ to $\IMMTI$} \label{algorithm: reduction from equivalence to multilinear equivalence for trimm}
\begin{algorithmic}
\State INPUT: Blackbox access to an $n$-variate degree-$d$ polynomial $f$.
\State OUTPUT: $A' \in \GLNF$ and $w\in \N$ such that $h(\vecx) = f(A'\vecx)$ is a $d$-tensor in the variable sets $\vecx_0, \ldots, \vecx_{d-1}$ which is isomorphic to $\IMM_{w,d}$.
\end{algorithmic}
\begin{algorithmic}[1]
\Statex \begin{center}\textcolor{gray}{Compute the irreducible $\G_f$-invariant subspaces } \end{center}
\State Call Algorithm \ref{alg:irr_inv_sub_space_gf} on input $f$. Let $\{\CAL{V}_0, \ldots, \CAL{V}_{d-1}\}$ be the spaces returned by Algorithm \ref{alg:irr_inv_sub_space_gf}. If Algorithm \ref{alg:irr_inv_sub_space_gf} returns `No' then output `No'. 

~
\Statex \begin{center}\textcolor{gray}{ Reduction to $\IMMTI$ }\end{center}
\State Call Algorithm~\ref{alg:layer_spaces} on input $\{\CAL{V}_0, \ldots, \CAL{V}_{d-1}\}$, and let $A' \in \GL(n,\F)$ and $w\in \N$ be the output of Algorithm \ref{alg:layer_spaces}. If Algorithm \ref{alg:layer_spaces} returns `No' then output `No'. Otherwise, return $A'$ and $w$. 
%
\end{algorithmic}
\end{algorithm}
Algorithm \ref{algorithm: reduction from equivalence to multilinear equivalence for trimm} is analysed by assuming that there is an $A\in \GL(n,\F)$ satisfying $f=\IMM_{w,d}(A\vecx)$. The final PIT at the end of Algorithm \ref{algorithm: equivalence testing for trimm} handles the case when $f$ is not equivalent to $\IMM_{w,d}$. 
In Step $1$, Algorithm \ref{alg:irr_inv_sub_space_gf} computes a set of bases of the irreducible $\G_f$-invariant subspaces. Algorithm \ref{alg:layer_spaces} in Step $2$ uses the bases to compute an $A' \in \GL(n,\F)$ and the $w\in \N$ such that $h(\vecx) = f(A'\vecx)$ is a $d$-tensor in the variable sets $\vecx_0, \ldots, \vecx_{d-1}$ which is isomorphic to $\IMM_{w,d}$. 
%
%
%
\subsection{Computing the irreducible $\Gf$-invariant subspaces} 
Algorithm \ref{alg:irr_inv_sub_space_gf} is similar to Algorithm 3 in \cite{KayalNST19} which computes the irreducible invariant subspaces of the Lie algebra of a polynomial equivalent to $\textnormal{IMM}_{w,d}$. 
\begin{algorithm}
\caption{Computing the irreducible $\mathfrak{g}_f$-invariant subspaces}\label{alg:irr_inv_sub_space_gf}
INPUT: Blackbox access to an $n$-variate degree-$d$ polynomial $f$.\\
OUTPUT: A set of bases of the irreducible $\mathfrak{g}_f$-invariant subspaces.
\begin{algorithmic}[1] 
\State Compute a basis $\{F_1, \ldots, F_a\}$ of $\mathfrak{g}_f$ using Lemma 2.2  in \cite{KayalNST19}.
\State Pick a random element $R' = \sum_{i=1}^{a}r_i F_i \in \mathfrak{g}_f$, where $r_i \in_r S$ is chosen independently and uniformly at random from $S \subseteq \F$ for every $i\in [n-1]$, and $|S| = 2n^3$ .
\State Compute the characteristic polynomial $q(x)$ of $R'$.
\State If $q(x)$ is not square-free then output `No'. Otherwise compute the irreducible factors of $q(x)$ over $\mathbb{F}$. Call the irreducible factors $p_1(x), \ldots, p_s(x)$. 

~
\State Compute bases of the null spaces $\CAL{N}_1', \ldots, \CAL{N}_s'$ of $p_1(R'), \ldots, p_s(R')$ respectively.
\State For every $i \in [s]$, pick a non-zero vector $\mathbf{v} \in \CAL{N}_i'$ and compute a basis of the closure of $\mathbf{v}$ under the action of $\mathfrak{g}_f$ using Algorithm 4 in \cite{KayalNST19}.
\State Let $\CAL{V}_1, \ldots, \CAL{V}_s$ be the list of the closure spaces (here, we are identifying spaces with their bases). Remove duplicates from the list by comparing every pair of spaces and get the pruned list $\CAL{V}_0, \ldots, \CAL{V}_{d-1}$. If the number of distinct closure spaces is not equal to $d$, or the dimension of all the closure spaces are not the same then output `No'. 
Else, output the list $\{\CAL{V}_0, \ldots, \CAL{V}_{d-1}\}$.
\end{algorithmic}
\end{algorithm}
\vspace{0.1in}

\textbf{Steps 1--4}: A basis of $\G_f$ is computed using Lemma 2.2  in \cite{KayalNST19} (also see \cite{Kayal12}). At Step 2, let $R \in \GIMM$ such that $R = A\cdot R' \cdot A^{-1}$. Since the matrices in $\GIMM$ are block-diagonal (Claim~\ref{claim: gtrimm is block diagonal}), $R$ is a block-diagonal matrix with individual blocks $R_0, \ldots, R_{d-1}$ as shown in Figure \ref{fig:R}. The characteristic polynomial $q(x)$ computed at Step 3 is square-free with high probability (Lemma~\ref{lemma:char_poly_sq_free}). Note that $q(x) = \prod_{k \in [0,d-1]} q_k(x)$, where $q_k(x)$ is the characteristic polynomial of $R_k$. At Step $4$, the algorithm invokes a univariate polynomial factorization algorithm over $\F$. Observe that every irreducible factor $p_i(x)$ of $q(x)$ is a factor of $q_k(x)$ for some $k \in [0,d-1]$.  
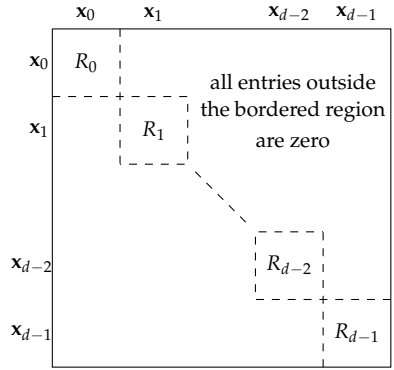
\begin{figure}[h]
\centering
\input{figures/random_matrix.tex}
\caption{A random matrix $R \in \GIMM$}
\label{fig:R}
\end{figure}
\vspace{0.1in}

\textbf{Step 5--7}: Let $\CAL{N}_i$ and $\CAL{N}_i'$ be the null spaces of $p_i(R)$ and $p_i(R')$ respectively. Then $\CAL{N}_i = A\CAL{N}_i'$. 
\begin{lemma}{\label{lemma:null_space_closure_gives_irred_spaces}}
Let $p_i(x)$ be an irreducible factor of $q_k(x)$, and $\mathbf{v} \in \CAL{N}_i'$ be a non-zero vector. Then, the closure of $\mathbf{v}$ under the action of $\mathfrak{g}_f$ is the irreducible $\G_f$-invariant subspace $A^{-1}\CAL{U}_k$. Thus, at the end of Step 7 there is a permutation $\sigma$ of $[0,d-1]$ such that $\CAL{V}_k = A^{-1}\CAL{U}_{\sigma(k)}$ for all $k\in [0,d-1]$. 
\end{lemma}
\begin{proof}
Consider the following claim. 
\begin{claim}{\label{claim:Ni-subset_uk}}
 $\CAL{N}_i' \subseteq A^{-1}\CAL{U}_k$.
\end{claim}
The proofs of Lemma \ref{lemma:irreducible_space_limm} and Corollary \ref{corollary:irr_inv_sub_f} in fact show that no $\G_f$-invariant subspace is properly contained in $A^{-1}\CAL{U}_k$. Observe that the closure of a vector under the action of $\G_f$ is a $\G_f$-invariant subspace by definition. Hence, by the above claim, the closure of $\mathbf{v}$ under the action of $\mathfrak{g}_f$ is $A^{-1}\CAL{U}_k$. \\
%

\noindent \emph{Proof of Claim \ref{claim:Ni-subset_uk}.} 
It is sufficient to show that $\CAL{N}_i \subseteq \CAL{U}_k$. 
Let $\mathbf{u} \in \CAL{N}_i$. Let $\mathbf{u}_{\ell} \in \F^{w^2}$ be the vector obtained by restricting $\vecu$ to the entries that are indexed by $\vecx_{\ell}$ variables for $\ell\in [0,d-1]$. 
The matrix $q_k(R)$ is block-diagonal with blocks $q_k(R_0), \ldots, q_k(R_{d-1})$. Since $\mathbf{u} \in \CAL{N}_i$, $p_i(R)\cdot\mathbf{u} = 0$ and so $q_k (R)\cdot\mathbf{u} = 0$. Hence,
\begin{equation}\label{equation: the submatrix equation for q_k}
q_k(R_{\ell})\cdot \mathbf{u_{\ell}} = 0 ~~~~~ \text{for all $\ell \in [0,d-1]$.}
\end{equation}
Further,
\begin{equation}\label{equation: the submatrix equation for q_ell}
q_{\ell}(R_{\ell})\cdot \mathbf{u_{\ell}}=0 ~~~~~\text{for all $\ell \in [0,d-1]$,}
\end{equation}
as $q_{\ell}(R_{\ell}) = 0$ (the characteristic polynomial of $R_{\ell}$ being $q_{\ell}(x)$). Since $q_k(x)$ and $q_{\ell}(x)$ are co-prime for $k \neq \ell$, there are polynomials $s(x)$ and $t(x)$ such that $s(x)q_k(x) + t(x)q_{\ell}(x) = 1$. This implies $s(R_{\ell})q_k(R_{\ell}) + t(R_{\ell})q_{\ell}(R_{\ell}) = I_{w^2}$. Hence, $s(R_{\ell})q_k(R_{\ell})\vecu_{\ell} + t(R_{\ell})q_{\ell}(R_{\ell})\vecu_{\ell} = \vecu_{\ell}$. From Equations \ref{equation: the submatrix equation for q_k} and \ref{equation: the submatrix equation for q_ell}, $\vecu_{\ell} = \mathbf{0}$ for all $k \neq \ell$.
\end{proof}
%
\subsection{Reduction to $\IMMTI$}
\begin{algorithm}
\caption{Reduction to $\IMMTI$}\label{alg:layer_spaces}
INPUT: The irreducible $\G_f$-invariant subspaces $\CAL{V}_0, \ldots, \CAL{V}_{d-1}$.\\
OUTPUT: $A' \in \GLNF$ and $w\in \N$ such that $h(\vecx) = f(A'\vecx)$ is a $d$-tensor in the variable sets $\vecx_0, \ldots, \vecx_{d-1}$ which is isomorphic to $\IMM_{w,d}$.
\begin{algorithmic}[1]
\State Determine $w$ such that $w^2$ is the dimension of each of the spaces  $\CAL{V}_0, \ldots, \CAL{V}_{d-1}$. If there does not exist such a $w$ then output `No'.
\State Construct the $n \times n$ matrix $V$ such that the $kw^2 + 1, \ldots , (k+1)w^2$ columns of $V$ are the basis vectors of $\CAL{V}_k$ for $k\in [0,d-1]$.

~
\State Compute a permutation $\tau$ of $[0,d-1]$ that is equal to $\sigma^{-1}$ (up to a ``rotation''), where $\sigma$ is the permutation in Lemma \ref{lemma:null_space_closure_gives_irred_spaces}. 
\State Compute a block-permuted permutation matrix $B$ that maps the variables in $\vecx_{\tau(k)}$ to the variables in $\vecx_k$ for all $k\in [0,d-1]$, i.e., $B \cdot (\vecx_{\tau(0)}~\vecx_{\tau(1)} \ldots \vecx_{\tau(d-1)})^T = (\vecx_{0}~\vecx_{1} \ldots \vecx_{d-1})^T$.
\State Return $A' = V\cdot B$ and $w$.
\end{algorithmic}
\end{algorithm}
\begin{definition}[Evaluation dimension \cite{ForbesS13, Nisan91}]\label{definition: evaldim}
Let $g(\vecx)$ be an $n$-variate polynomial and $\vecx' \subseteq \vecx$. Let $g(\vecx)_{\vecx' = \boldmath\alpha}$ denote the partial evaluation of $g$ at $\vecx' = \boldmath\alpha \in \F^{|\vecx'|}$. The evaluation dimension of $g$ with respect to $\vecx'$ is defined as $\textnormal{Evaldim}_{\vecx'}(g) := \textnormal{dim}(\textnormal{span}_{\F}(\{g(\vecx)_{\vecx' = \boldmath\alpha}: \boldmath\alpha \in \F^{|\vecx'|}\} ))\ \, .$
\end{definition}

We use the above definition to analyse Algorithm \ref{alg:layer_spaces}. \vspace{0.1in}

\textbf{Steps 1--2}: The correctness of Step 1 follows from Corollary \ref{corollary:irr_inv_sub_f}. Let $V_k$ be the $n \times w^2$ matrix whose columns are the basis vectors of the $\G_f$-invariant subspace $\mathcal{V}_k$. Then the matrix $V$ constructed at Step 2 is obtained by concatenating the matrices $V_0, \ldots, V_{d-1}$ in this order, denoted $V_0 | V_1 | \ldots | V_{d-2} | V_{d-1}$. From Lemma \ref{lemma:null_space_closure_gives_irred_spaces}, there is a permutation $\sigma$ of $[0,d-1]$ such that $\CAL{V}_k = A^{-1}\CAL{U}_{\sigma(k)}$. Hence, there is a matrix $E_k\in \F^{n\times w^2}$ such that $V_k = A^{-1}E_k$ and the non-zero entries of $E_k$ are confined to the rows indexed by $\vecx_{\sigma(k)}$ variables. Let $E = E_0 | \ldots | E_{d-1}$. Then $V= A^{-1}\cdot E$. Observe that $E$ is a block-permuted matrix, i.e., the columns indexed by $\vecx_k$ variables have non-zero entries confined to the rows indexed by $\vecx_{\sigma(k)}$ variables. Thus, $g(\vecx) := f(V\vecx) = \IMM_{w,d}(E\vecx)$. \vspace{0.1in}  

\textbf{Steps 3--5}: Step $3$ uses the algorithm in next claim to determine $\tau$.  
\begin{claim}{\label{claim:reordering_layer_spaces}}
There is a randomized polynomial-time algorithm that takes input blackbox access to $g$ and with probability $1-o(1)$ outputs a permutation $\tau$ of $[0,d-1]$ such that there is an $\ell\in [0,d-1]$ satisfying either a) $\tau(k) = \sigma^{-1}(\ell+ k)$ for all $k\in [0,d-1]$, or b) $\tau(k) = \sigma^{-1}(\ell - k)$ for all $k\in [0,d-1]$.
\end{claim}
The claim is proved below after completing the analysis of Algorithm \ref{alg:layer_spaces}. Assume that $\tau(k) = \sigma^{-1}(\ell+ k)$ for all $k\in [0,d-1]$; the analysis for $\tau(k) = \sigma^{-1}(\ell - k)$ for all $k\in [0,d-1]$ is similar. Since $g = \IMM_{w,d}(E\vecx)$, there is a full-rank $(w,d,n)$-set-multilinear matrix product $X_0\cdots X_{d-1}$  in the variable sets $\vecx_{\sigma^{-1}(0)}, \ldots, \vecx_{\sigma^{-1}(d-1)}$ respectively such that $g(\vecx) = \text{tr}(X_0\cdots X_{d-1}) = \text{tr}(X_{\ell}\cdot X_{\ell+1}\ldots X_{d-1}\cdot X_0\ldots X_{\ell-1})$. Renaming $X_{\ell+k}$ as $X_k$ for all $k\in [0,d-1]$ and reusing symbols, it is inferred that there is a full-rank $(w,d,n)$-set-multilinear matrix product $X_0\cdots X_{d-1}$  in the variable sets $\vecx_{\tau(0)}, \ldots, \vecx_{\tau(d-1)}$ respectively such that $g = \text{tr}(X_0\cdots X_{d-1})$. Hence, at Steps $4$ and $5$, it is readily seen that $g(B\vecx) = f(VB\vecx)$ is computed by a full-rank $(w,d,n)$-set-multilinear matrix product $X'_0\cdots X'_{d-1}$ in the variable sets $\vecx_{0}, \ldots, \vecx_{d-1}$ respectively, i.e., $f(VB\vecx) = \text{tr}(X'_0\cdots X'_{d-1})$. \vspace{0.1in}


\begin{proof}[Proof of Claim \ref{claim:reordering_layer_spaces}]
The following observation is the key to computing $\tau$.
\begin{observation}{\label{obs:sigma_compute}}
Let $\ell \in [0,d-1]$ and $r=\sigma^{-1}(\ell)$. Then a) for $r' \in \{\sigma^{-1}(\ell-1), \sigma^{-1}(\ell+1)\}$, $\textnormal{Evaldim}_{\vecx_{r} \uplus \vecx_{r'}}(g) = w^2$, and b) for $r' \in [0,d-1]\setminus \{\sigma^{-1}(\ell), \sigma^{-1}(\ell+1), \sigma^{-1}(\ell-1)\}$, $\textnormal{Evaldim}_{\vecx_{r} \uplus \vecx_{r'}}(g) = w^4$. There is a randomized polynomial-time algorithm to compute $\textnormal{Evaldim}_{\vecx_{r} \uplus \vecx_{r'}}(g)$ for all $r, r' \in [0, d-1]$.
\end{observation}
The observation is proved after completing the proof of the claim. Observation \ref{obs:sigma_compute} is used ${d \choose 2}$ times to determine $S_{r} = \{\sigma^{-1}(\ell-1), \sigma^{-1}(\ell+1)\}$ where $r = \sigma^{-1}(\ell)$, for every $r \in [0,d-1]$. Using the knowledge of $S_{0}, \ldots , S_{d-1}$, $\tau$ is determined which is equal to $\sigma^{-1}$ up to a rotation as follows. Choose an arbitrary element $r\in [0,d-1]$ and set $\tau(0) = r$, and let $\ell \in [0,d-1]$ be such that $\sigma^{-1}(\ell) = r$.  We can construct $\tau$ by choosing either of the elements in $S_{r} = \{\sigma^{-1}(\ell-1), \sigma^{-1}(\ell+1)\}$: if $\sigma^{-1}(\ell-1)$ is chosen then $\tau$ constructed will be such that $\tau(k) = \sigma^{-1}(\ell-k)$ for $k\in [0,d-1]$, and if $\sigma^{-1}(\ell+1)$ is chosen then $\tau$ constructed will be such that $\tau(k) = \sigma^{-1}(\ell+k)$ for $k\in [0,d-1]$. Without loss of generality assume $\sigma^{-1}(\ell+1)$ is chosen. Set $\tau(1) = \sigma^{-1}(\ell+1)$. The remaining part of $\tau$ is determined sequentially as follows. Suppose, for some $k \in [0,d-2]$, $\tau(i) = \sigma^{-1}(\ell+i)$ for all $i\in [0,k]$. Then $S_{\sigma^{-1}(\ell+k)} = \{\sigma^{-1}(\ell+k-1), \sigma^{-1}(\ell+k+1)\}$ and $\tau(k-1) = \sigma^{-1}(\ell+k-1)$. Choose the other element in $S_{\sigma^{-1}(\ell+k)}$ and set $\tau(k+1) = \sigma^{-1}(\ell+k+1)$. 
\end{proof}

\begin{proof}[Proof of Observation~\ref{obs:sigma_compute}]
There is a full-rank $(w,d,n)$-set-multilinear matrix product $X_0\cdots X_{d-1}$  in the variable sets $\vecx_{\sigma^{-1}(0)}, \ldots, \vecx_{\sigma^{-1}(d-1)}$ respectively such that $g(\vecx) = \text{tr}(X_0\cdots X_{d-1})$. 

\emph{Case a}:~ Suppose $r' = \sigma^{-1}(\ell+1)$; the proof for $r' = \sigma^{-1}(\ell-1)$ is similar. Let $\CAL{A} = \textnormal{span}_{\F}\{g(\mathbf{x})|_{\mathbf{x}_{r} \uplus \mathbf{x}_{r'} = {\boldmath\alpha} },$ $\boldmath\alpha \in \F^{2w^2}\}$. 
Observe that $g(\vecx) = \textnormal{tr}(X_0\ldots X_{d-1}) = \text{tr}(X_{\ell}\cdot X_{\ell+1}\ldots X_{d-1}\cdot X_0\ldots X_{\ell-1})$. Let $Y_1$ be the row vector of size $w^2$ whose $((i-1)\cdot w+j)$-th entry is the $(i,j)$-th entry of $X_{\ell}\cdot X_{\ell+1}$, for $i,j\in [w]$. Similarly, let $Y_2$ be the column vector of size $w^2$ whose $((j-1)\cdot w+i)$-th entry is the $(i,j)$-th entry of $X_{\ell+2}\ldots X_{d-1}\cdot X_0\ldots X_{\ell-1}$, for $i,j\in [w]$. From the construction of $Y_1$ and $Y_2$, $g(\vecx) = Y_1\cdot Y_2$. Since $X_{\ell}, X_{\ell+1}$ are full-rank linear matrices in disjoint variable sets, there is a point $\alpha_i \in \F^{2w^2}$ such that the $i$-th entry of $Y_1$ evaluated at this point is equal to one and the remaining entries are zero, for every $i\in [w^2]$. Hence, every entry of $Y_2$ is in $\CAL{A}$, and further as $X_{\ell+1}\ldots X_{d-1}\cdot X_{0}\ldots X_{\ell-1}$ is a full-rank set-multilinear matrix product, the $w^2$ entries of $Y_2$ are $\F$-linearly independent. Thus the entries of $Y_2$ form a basis of $\CAL{A}$, and $\textnormal{Evaldim}_{\vecx_{r} \uplus \mathbf{x}_{r'}}(g) = \textnormal{dim}(\CAL{A}) = w^2$. \vspace{0.1in}

\emph{Case b}:~ Suppose $r' = \sigma^{-1}(\ell')$ and $\ell' \notin \{\ell-1,\ell, \ell+1\}$. Let $\CAL{A} = \textnormal{span}_{\F}\{g(\mathbf{x})|_{\mathbf{x}_{r} \uplus \mathbf{x}_{r'} = {\alpha} }, \alpha \in \F^{2w^2}\}$. Again $g(\vecx) = \textnormal{tr}(X_{\ell}\cdot X_{\ell+1}\ldots X_{\ell'}\ldots X_{\ell-1})$. Let $P = X_{\ell+1} \ldots X_{\ell'-1} = (p_{i,j})_{i,j\in [w]}$ and $T = X_{\ell'+1}\ldots X_{\ell-1} = (t_{i,j})_{i,j\in [w]}$. Since $X_0\ldots X_{d-1}$ is a full-rank set-multilinear matrix product, the $w^4$ polynomials $\{p_{i_1,j_1} \cdot t_{i_2,j_2} \mid i_1,j_1, i_2,j_2 \in [w]\}$ are linearly independent over $\F$. Moreover,  $\vecx_{r} \uplus \vecx_{r'}$ can be substituted appropriately such that these $w^4$ polynomials are in $\CAL{A}$. Since $\CAL{A} =  \textnormal{span}_{\F}\{ p_{i_1,j_1} \cdot t_{i_2,j_2} \mid i_1,j_1, i_2,j_2 \in [w]\}$, the $w^4$ polynomials $\{p_{i_1,j_1} \cdot t_{i_2,j_2} \mid i_1,j_1, i_2,j_2 \in [w]\}$ form a basis of $\CAL{A}$. This implies $\textnormal{Evaldim}_{\vecx_{r} \uplus \mathbf{x}_{r'}}(g) = \textnormal{dim}(\CAL{A}) = w^4$. \vspace{0.1in}

\emph{A polynomial-time randomized procedure to compute} $\textnormal{Evaldim}_{\mathbf{x}_{r} \uplus \mathbf{x}_{r'}}(g)$:  Let $S\subset \F$ such that $|S| = n^{4}$. Choose points $\mathbf{a}_1, \ldots, \mathbf{a}_{w^4} \in_{r} S^{2w^2}$ independently and uniformly at random and output the dimension of the $\F$-linear space spanned by the polynomials $g(\vecx\setminus \{\vecx_r, \vecx_{r'}\}, \mathbf{a}_1), \ldots, g(\vecx\setminus \{\vecx_r, \vecx_{r'}\}, \mathbf{a}_{w^4})$ using Claim 2.2 in \cite{KayalNST19}. The proof of correctness of this procedure is similar to the proof of correctness of the randomized procedure in Observation E.1 in \cite{KayalNST19}. 
\end{proof}

%% file: figures/random_matrix.tex
\begin{tikzpicture}
\coordinate (a) at (0,0);
\coordinate (b) at ($(a) + (4.5,0)$);
\coordinate (c) at ($(a) + (4.5,-4.5)$);
\coordinate (d) at ($(a) + (0,-4.5)$);

\node at ($ (a) + (0.45,-0.45)$) [fill=white!100!] {\scriptsize $R_0$};
\node at ($ (a) + (1.35,-1.35)$) [fill=white!100!] {\scriptsize $R_1$};

\node at ($ (c) - (1.35,-1.35)$) [fill=white!100!] {\scriptsize $R_{d-2}$};
\node at ($ (c) - (0.45,-0.45)$) [fill=white!100!] {\scriptsize $R_{d-1}$};

\node at ($ (a) + (-0.17,-0.45)$) [fill=white!100!] {\scriptsize $\vecx_0$};
\node at ($ (a) + (-0.17,-1.35)$) [fill=white!100!] {\scriptsize $\vecx_1$};

\node at ($ (d) + (-0.27,+1.35)$) [fill=white!100!] {\scriptsize $\vecx_{d-2}$};
\node at ($ (d) + (-0.27,+0.45)$) [fill=white!100!] {\scriptsize $\vecx_{d-1}$};

\node at ($ (a) + (0.45,0.2)$) [fill=white!100!] {\scriptsize $\vecx_0$};
\node at ($ (a) + (1.35,0.2)$) [fill=white!100!] {\scriptsize $\vecx_1$};

\node at ($ (b) + (-1.35,0.2)$) [fill=white!100!] {\scriptsize $\vecx_{d-2}$};
\node at ($ (b) + (-0.45,0.2)$) [fill=white!100!] {\scriptsize $\vecx_{d-1}$};

\node at ($ (a) + (3.2,-0.7)$) [fill=white!100!] {\scriptsize all entries outside};
\node at ($ (a) + (3.2,-1.1)$) [fill=white!100!] {\scriptsize the bordered region};
\node at ($ (a) + (3.2,-1.5)$) [fill=white!100!] {\scriptsize are zero};

\draw (a) -- (b) -- (c) -- (d) -- (a);
\draw [dashed] ($(a) + (0.9,0)$) -- ($(a) + (0.9,-1.8)$);
\draw [dashed] ($(a) + (0,-0.9)$) -- ($(a) + (1.8,-0.9)$) -- ($(a) + (1.8,-1.8)$) -- ($(a) + (0.9,-1.8)$);
\draw [dashed] ($(a) + (1.9,-1.9)$) -- ($(c) - (1.9,-1.9)$);
\draw [dashed] ($(c) - (0.9,0)$) -- ($(c) - (0.9,-1.8)$);
\draw [dashed] ($(c) - (0,-0.9)$) -- ($(c) - (1.8,-0.9)$) -- ($(c) - (1.8,-1.8)$) -- ($(c) - (0.9,-1.8)$);

\end{tikzpicture} 

%% file: secappendix_reduction_to_det_eq.tex
\section{Proofs from Section \ref{sec: reduction from IMM to DET}} \label{secappendix: reduction to det}
\emph{Claim \ref{claim: uniqueness from transpose} (restated): Let $X$ be a $w\times w$ full-rank linear matrix and $Y=I_w\otimes X$. Then there does not exist non-zero matrices $T,S \in \CAL{M}_{w^2}(\F)$ such that $T\cdot Y = Y^T\cdot S$.}
\begin{proof}
Since $X$ is a full-rank linear matrix, by applying an invertible transformation we may assume without loss of generality that the entries of $X$ are distinct $w^2$ variables. Hence, it is sufficient to prove the claim when $X$ is symbolic matrix with entries being distinct variables. Suppose for contradiction, there are  non-zero matrices $T$ and $S$ such that $T,S \in \CAL{M}_{w^2}(\F)$ and $T\cdot Y = Y^T\cdot S$. Let $T_{i,j}$ (respectively $S_{i,j}$) denote the $(i,j)$-th $w\times w$ sub-matrix of $T$ (respectively $S$) corresponding to the rows numbered from $(w(i-1) + 1)$ to $(wi)$, and columns numbered from $(w(j-1)+1)$ to $(wj)$ of $T$ (respectively $S$), for $i,j\in [w]$. Then $T_{i,j}\cdot X = X^T\cdot S_{i,j}$, for every $i,j\in [w]$. For $u\in [2,w]$, observe that the $(1,u)$ entries of $T_{i,j}\cdot X$ and $X^T\cdot S_{i,j}$ are variable disjoint implying that all the columns except the first column of of $S_{i,j}$ are zero columns for every $i,j\in [w]$. Similarly comparing the $(2,1)$ entries of $T_{i,j}\cdot X$ and $X^T\cdot S_{i,j}$, it is observed that even the first column of $S_{i,j}$ is a zero column for every $i,j\in [w]$. This implies $S$ is a zero matrix, and hence $T$ is a zero matrix.
\end{proof}
\emph{Claim \ref{claim: Uniqueness of commuting matrices} (restated): Let $X$ be a $w \times w$ full-rank linear matrix and $Y = I_w\otimes X$, and suppose $T,S \in \CAL{M}_{w^2}(\F)$ such that $T\cdot Y = Y \cdot S$. Then $T = S = M \otimes I_w$ for some $M \in \CAL{M}_w(\F)$.}
\begin{proof}
Similar to the proof of Claim \ref{claim: uniqueness from transpose}, it is sufficient to  prove Claim \ref{claim: Uniqueness of commuting matrices} for the case when $X$ is a  $w\times w$ symbolic matrix with entries being distinct variables. Let $T,S \in \CAL{M}_{w^2}(\F)$ be such that $T\cdot Y = Y\cdot S$. Also let $\veca \in \F^{w^2}$ be such that  $X$ evaluated at $\veca$ is equal to $I_{w}$. Now evaluating the expression $T\cdot Y= Y\cdot S$ at $\veca$, it is  inferred that $T=S$. Let $T_{i,j}$ denote the $(i,j)$-th $w\times w$ sub-matrix of $T$ corresponding to the rows numbered from $(w(i-1) + 1)$ to $(wi)$, and columns numbered from $(w(j-1)+1)$ to $(wj)$ of $T$, for $i,j\in [w]$. Then $T_{i,j}\cdot X = X\cdot T_{i,j}$, for every $i,j\in [w]$. Since the entries of $X$ are distinct variables, $T_{i,j} = m_{i,j}I_w$, where $m_{i,j}\in \F$. Hence $T = M\otimes I_{w}$, where $M = (m_{i,j})_{i,j\in [w]}$.
%
\end{proof}
\emph{Observation \ref{observation:exactly_one_succeds} (restated): If $h(\vecx_0, \ldots, \vecx_{d-1})$ is isomorphic to $\IMM_{w,d}$ then for matrices $Y'_k$ and $Z_k$ as computed in Algorithm \ref{algorithm: multilinear equivalence test for trimm}, where $k \in [1,d-2]$, there are no matrices $T'_{k-1}, S'_k\in \GL(w^2,\F)$ such that both $T'_{k-1} \cdot Y'_{k} = Z_k \cdot S'_k$ and $T'_{k-1} \cdot Y'_{k} = Z_k^{T} \cdot S'_k$ are simultaneously true.}
\begin{proof}
Since $h(\vecx)$ is multilinearly equivalent to $\IMM_{w,d}$, as argued in Section \ref{section: multilinear equivalence testing} for all $k \in [1,d-2]$, $Z_k = I_w \otimes X'_k$, and $Y'_k = T_{k-1}^{-1}\cdot Y_{k}\cdot T_k$. Moreover, either 
$$Y_k = (I_w\otimes C_k)\cdot Z_k\cdot (I_w\otimes D_k) ] ~~~~\text{or}~~~~ Y_k = (I_w\otimes C_k)\cdot Z_k^T\cdot (I_w\otimes D_k).$$ 
We prove the observation when $Y_k = (I_w\otimes C_k)\cdot Z_k\cdot (I_w\otimes D_k)$, and the proof for $Y_k = (I_w\otimes C_k)\cdot Z_k^T\cdot (I_w\otimes D_k)$ is similar. Suppose there are matrices $T'_{k-1}, S'_k\in \GL(w^2,\F)$ such that both $T'_{k-1} \cdot Y'_{k} = Z_k \cdot S'_k$ and $T'_{k-1} \cdot Y'_{k} = Z_k^{T} \cdot S'_k$ are simultaneously true. Then Equations \ref{equation: multiplication without transpose} and \ref{equation: multiplication with transpose} are simultaneously true. 
\begin{align}
    T'_{k-1} \cdot (T_{k-1}^{-1}\cdot Y_{k}\cdot T_k) &= Z_k \cdot S'_k \nonumber\\
    (T'_{k-1}T_{k-1}^{-1})\cdot (I_w\otimes C_k)\cdot Z_k\cdot (I_w\otimes D_k)\cdot T_k &= Z_k \cdot S'_k  \nonumber\\
    (T'_{k-1}T_{k-1}^{-1})\cdot (I_w\otimes C_k)\cdot Z_k &= Z_k \cdot (S'_k T_{k}^{-1}) \cdot (I_w\otimes D_k^{-1}) \label{equation: multiplication without transpose}
\end{align}
\begin{align}
    T'_{k-1} \cdot (T_{k-1}^{-1}\cdot Y_{k}\cdot T_k) &= Z_k^T \cdot S'_k \nonumber\\
    (T'_{k-1}T_{k-1}^{-1})\cdot (I_w\otimes C_k)\cdot Z_k\cdot (I_w\otimes D_k)\cdot T_k &= Z_k^T \cdot S'_k \nonumber\\
    (T'_{k-1}T_{k-1}^{-1})\cdot (I_w\otimes C_k)\cdot Z_k &= Z_k^T \cdot (S'_k T_{k}^{-1}) \cdot (I_w\otimes D_k^{-1}) \label{equation: multiplication with transpose}
\end{align}
Since $X'_k$ is a full-rank linear matrix in $\vecx_k$ variables and $Z_k = I_w\otimes X'_k$, this contradicts Claim \ref{claim: uniqueness from transpose}.
\end{proof}
\emph{Observation \ref{observation:uniquess_ti_si} (restated):  The matrices $T'_{k-1}$ and $S'_k$ computed at Step 5 of Algorithm \ref{algorithm: multilinear equivalence test for trimm}, where $k\in [1,d-2]$, satisfy the following: $(T'_{k-1})^{-1} = T_{k-1}^{-1} \cdot (I_w \otimes C_k)  \cdot (M_k^{-1} \otimes I_w)$ and $S'_k = (M_k \otimes I_w) \cdot (I_w \otimes D_k) \cdot T_k$, where $M_k \in \GL(w, \F)$. }
\begin{proof}
Substitute $Y'_k = T_{k-1}^{-1} \cdot (I_w \otimes C_k) \cdot Z_k \cdot (I_w \otimes D_k) \cdot T_k$ \,in $T'_{k-1} \cdot Y'_k = Z_k \cdot S'_k$. Then
$$T'_{k-1} \cdot T_{k-1}^{-1} \cdot (I_w \otimes C_k) \cdot Z_k = Z_k \cdot  S'_k \cdot T_k^{-1} \cdot  (I_w \otimes D_k^{-1}) .$$
Hence, from Claim \ref{claim: Uniqueness of commuting matrices}, there is a matrix $M_k\in \GL(w,\F)$ such that
$$ (T'_{k-1} T_{k-1}^{-1}) \cdot (I_w \otimes C_k) = (S'_k T_k^{-1}) \cdot (I_w \otimes D_k^{-1}) = M_k \otimes I_w \, .$$
This implies $(T'_{k-1})^{-1} = T_{k-1}^{-1} \cdot (I_w \otimes C_k)  \cdot (M_k^{-1} \otimes I_w)$ and $S'_k = (M_k \otimes I_w) \cdot (I_w \otimes D_k) \cdot T_k$. 
\end{proof}
\emph{Observation \ref{obs:structure of widehatY ABP} (restated): Let $M_{1}, \ldots, M_{d-2}$ be the matrices as defined in Observation \ref{observation:uniquess_ti_si}. Then 
\begin{enumerate}
\item  $\widehat{Y}_k = (M_k M_{k+1}^{-1} \otimes I_w)\cdot(I_w \otimes (C_k^{-1} \cdot X_k \cdot C_{k+1}))$ ~~for $k \in [1,d-3]$, 
\item  $\widehat{Y}_{d-2} = I_w \otimes (C_{d-2}^{-1} \cdot X_{d-2} \cdot D_{d-2}^{-1})$, 
\item  $\widehat{Y}_0 = Y_0\cdot (I_w \otimes C_1)  \cdot (M_1^{-1} \otimes I_w)$, and $\widehat{Y}_{d-1} = (M_{d-2} \otimes I_w) \cdot (I_w \otimes D_{d-2})\cdot Y_{d-1}$.
\end{enumerate}}
\begin{proof}
Let $T'_k, S'_k, T_k$ for $k\in [0,d-2]$, and $Y'_k, Y_k, X'_k, X_k$ for $k\in [0,d-1$ be as in Section \ref{section: multilinear equivalence testing}. \vspace{0.1in}

\emph{a)} For $k \in [1,d-3]$, $\widehat{Y}_k = T'_{k-1}\cdot Y'_k\cdot (T'_k)^{-1}$,~ $Y'_k = T_{k-1}^{-1}\cdot Y_k\cdot T_k$, ~and $Y_k = I_w\otimes X_k$. From Observation  \ref{observation:uniquess_ti_si}, 
$$T'_{k-1} =  (M_k \otimes I_w) \cdot (I_w \otimes C_k^{-1}) \cdot T_{k-1} ~~~~\text{and}~~~~(T'_{k})^{-1} = T_{k}^{-1} \cdot (I_w \otimes C_{k+1})  \cdot (M_{k+1}^{-1} \otimes I_w)~,$$ and hence for $k\in [1,d-3]$
$$\widehat{Y}_k = T'_{k-1}\cdot Y'_k\cdot (T'_k)^{-1} = (M_k \otimes I_w) \cdot (I_w \otimes C_k^{-1}) \cdot Y_k \cdot (I_w \otimes C_{k+1})  \cdot (M_{k+1}^{-1} \otimes I_w)\,.$$
Since $(I_w \otimes (C_k^{-1}\cdot X_k\cdot C_{k+1})) \cdot (M_{k+1}^{-1} \otimes I_w) = (M_{k+1}^{-1} \otimes I_w) \cdot (I_w \otimes (C_k^{-1} \cdot X_k \cdot C_{k+1}))$, for $k\in [1,d-3]$
$$\widehat{Y}_k = (M_k M_{k+1}^{-1} \otimes I_w)\cdot(I_w \otimes (C_k^{-1}\cdot  X_k \cdot C_{k+1}))\, .$$
\emph{b)} Recall that $\widehat{Y}_{d-2} = T'_{d-3}\cdot Y'_{d-2}\cdot (S'_{d-2})^{-1}$, $Y'_{d-2} = T_{d-3}^{-1}\cdot Y_{d-2}\cdot T_{d-2}$, and $Y_{d-2} = I_w\otimes X_{d-2}$. From Observation \ref{observation:uniquess_ti_si}, 
$$T'_{d-3} =  (M_{d-2} \otimes I_w) \cdot (I_w \otimes C_{d-2}^{-1}) \cdot T_{d-3}, ~~~~\text{and}~~~~(S'_{d-2})^{-1} = T_{d-2}^{-1}\cdot  (I_w \otimes D_{d-2}^{-1}) \cdot (M_{d-2}^{-1} \otimes I_w).$$
Hence,
$$\widehat{Y}_{d-2} = T'_{d-3}\cdot Y'_{d-2}\cdot (S'_{d-2})^{-1} = (M_{d-2} \otimes I_w) \cdot (I_w \otimes C_{d-2}^{-1}) \cdot Y_{d-2}\cdot (I_w \otimes D_{d-2}^{-1}) \cdot (M_{d-2}^{-1} \otimes I_w)\,.$$
Since $(I_w \otimes (C_{d-2}^{-1} \cdot X_{d-2} \cdot D_{d-2}^{-1})) \cdot (M_{d-2}^{-1} \otimes I_w) = (M_{d-2}^{-1} \otimes I_w)\cdot (I_w \otimes (C_{d-2}^{-1} \cdot X_{d-2} \cdot D_{d-2}^{-1}))$, 
$$\widehat{Y}_{d-2} = I_w \otimes (C_{d-2}^{-1} \cdot X_{d-2} \cdot D_{d-2}^{-1})\,.$$
\emph{c)} Recall $\widehat{Y}_0 = Y'_0\cdot (T'_0)^{-1}$, ~$\widehat{Y}_{d-1} = S'_{d-2}\cdot Y'_d$, and $Y'_0 = Y_0\cdot T_0$, ~$Y'{d-1} = T_{d-2}^{-1}\cdot Y_{d-1}$. From Observation \ref{observation:uniquess_ti_si}, $(T'_{0})^{-1} = T_0^{-1}\cdot (I_w \otimes C_1)  \cdot (M_1^{-1} \otimes I_w)$ and $S'_{d-2} = (M_{d-2} \otimes I_w) \cdot (I_w \otimes D_{d-2}) \cdot T_{d-2}$. Hence,
$$\widehat{Y}_0 = Y'_0\cdot (T'_0)^{-1} = Y_0\cdot (I_w \otimes C_1)  \cdot (M_1^{-1} \otimes I_w)~~\text{and}$$
$$ Y'_{d-1} = T_{d-2}^{-1}\cdot Y_{d-1}= (M_{d-2} \otimes I_w) \cdot (I_w \otimes D_{d-2}) \cdot Y_{d-1}\, .$$
\end{proof}

%% file: sec_reduction_from_mult_trace_to_deg3_mult_trace.tex
\section{Reduction from $\IMMTI$ to $\MMTI$:~ Proof of Theorem \ref{theorem: reduction from equivalence testing for degree d to degree 3}}\label{sec: reduction from degree d to degree 3}
The input to Algorithm \ref{algorithm: reduction from degree d to degree 3} is blackbox access to a $d$-tensor $f$ in the variable sets $\vecx_{0}, \ldots, \vecx_{d-1}$, and oracle access to $\MMTI$. With high probability the algorithm does the following: If $f$ is isomorphic to $\IMM_{w,d}$ then it outputs $B_0,\ldots, B_{d-1}\in \GL(w^2,\F)$ such that $f =\IMM_{w,d}(B_0\vecx_0, \ldots ,B_{d-1}\vecx_{d-1})$, otherwise it outputs `No'. Since a PIT at the end of the algorithm ensures that the output of the algorithm is correct with high probability, we can assume that $f$ is isomorphic to $\IMM_{w,d}$.
\begin{algorithm}[ht!]
\caption{Reduction from $\IMMTI$ to $\MMTI$} \label{algorithm: reduction from degree d to degree 3}
\begin{algorithmic}
\State INPUT: Blackbox access to a $d$-tensor $f(\vecx_0, \ldots, \vecx_{d-1})$, where $d\geq 3$, and oracle access to $\MMTI$.
\State OUTPUT: Matrices $B_0, \ldots, B_{d-1} \in \GL(w^2,\F)$ such that $f(\vecx) = \IMM_{w,d}(B_0\vecx_0, \ldots, B_{d-1}\vecx_{d-1})$.
\end{algorithmic}
\begin{algorithmic}[1]
\Statex \begin{center}\textcolor{gray}{Computing the first three matrices}\end{center}
\State Choose $d-3$ random points $\veca_{3}, \ldots \veca_{d-1} \in S^{w^2}$, where $S \subseteq \F$ and $|S| \geq n^5$. Let $\vecy = \uplus_{k\in [0,2]} \vecx_k$, and $h(\vecy) = f(\vecy, \veca_3, \ldots, \veca_{d-1})$.
\State Query $\MMTI$ on input $h(\vecy)$. If $\MMTI$ outputs `No' then output `No'. Otherwise, let $B_0,B_1,B_2 \in \GL(w^2, \F)$ be the output of $\MMTI$.  
Using $B_0,B_1,B_2$ compute $w\times w$ linear matrices $X'_0, X'_1, X'_2$ respectively such that $h(\vecy) = \tr(X'_0\cdot X'_1 \cdot X'_2)$.

~
\Statex \begin{center}\textcolor{gray}{The $d=4$ case}\end{center}
\State For $i,j \in [w]$, and $k \in [0,2]$, compute $\vecb_{i,j}^{(k)} \in \F^{w^2}$ such that $X'_k(\vecb_{i,j}^{(k)})$ has one in the $(i,j)$-th entry and zero elsewhere. Here $X'_k(\vecb_{i,j}^{(k)})$ is the matrix $X'_k$ with its entries evaluated at $\vecb_{i,j}^{(k)}$. 
\State If $d=4$ then construct $X'_3$ such that its $(i,j)$-th entry is the linear form $f(\vecb^{(0)}_{j,1},\vecb^{(1)}_{1,1},\vecb_{1,i}^{(2)},\vecx_{3})$. Let $B_3$ be the transformation matrix on $\vecx_3$ which is derived from $X'_3$. Go to Step 10.  

~
\Statex \begin{center}\textcolor{gray}{The $d \geq 5$ case}\end{center}
\State Let $g(\vecx\setminus \vecy) = f(\vecb^{(0)}_{1,1}, \vecb^{(1)}_{1,1}, \vecb^{(2)}_{1,1}, \vecx\setminus\vecy)$. Use the set-multilinear ABP reconstruction algorithm in \cite{KlivansS03} (also see Claim 2.4 in \cite{KayalNST19}) to construct a full-rank $(w,d,n)$-set-multilinear ABP $Y'_3 \cdots Y'_{d-1}$ in $\vecx_3, \ldots ,\vecx_{d-1}$ variables that computes the polynomial $g$. 
\State For $j\in [w]$, compute $\vecb_{j}^{(d-1)}\in \F^{w^2}$ such that the $j$-th entry of $Y'_{d-1}(\vecb_{j}^{(d-1)}) \in \F^w$ is one and other entries are zero. For $k\in [4,d-2]$ let $X'_k = Y'_k$, and compute $\vecb_{i,j}^{(k)} \in \F^{w^2}$ such that $X'_k(\vecb_{i,j}^{(k)}) \in \F^{w\times w}$ has one in the $(i,j)$-th entry and other entries are zero, where $i,j\in [w]$.

~
\State Construct $X'_3$ whose $(i,j)$-th entry is $f(\vecb_{1,1}^{(0)}, \vecb_{1,1}^{(1)}, \vecb_{1,i}^{(0)}, \vecx_3, \mathbf{b}_{j,j}^{(4)}, \mathbf{b}_{j,j}^{(5)}, \ldots, \mathbf{b}_{j}^{(d-1)})$. 
For $i,j\in [w]$, compute $\vecb_{i,j}^{(3)} \in \F^{w^2}$ such that the $(i,j)$-th entry of $X'_3(\vecb_{i,j}^{(3)})$ is one and other entries zero.
\State Construct $X'_{d-1}$ such that its $(i,j)$ entry is $f(\vecb_{j,1}^{(0)}, \vecb_{1,1}^{(1)}, \ldots, \vecb_{1,i}^{(d-2)}, \vecx_{d-1})$, for $i,j\in [w]$. Finally, $B_k$ be the transformation matrix on $\vecx_k$ which is derived from $X'_k$ for $k \in [3, d-1]$.

~
\Statex \begin{center} \textcolor{gray}{Final PIT} \end{center}
\State Pick random points $\mathbf{c}_0, \ldots, \mathbf{c}_{d-1} \in S^{w^2}$, where $S \subseteq \F$ and $|S| \geq n^5$. If $f(\mathbf{c}_0, \ldots, \mathbf{c}_{d-1}) = \IMM_{w,d}(B_0\mathbf{c}_0, \ldots, B_{d-1}\mathbf{c}_{d-1})$ then output $B_0, \ldots ,B_{d-1}$, otherwise output `No'. 
\end{algorithmic}
\end{algorithm}
\vspace{0.1in}

\textbf{Steps 1--2}: Let $n=w^2d$. Since $f$ is isomorphic to $\IMM_{w,d}$, there is a full-rank $(w,d,n)$-set-multilinear matrix product $X_0 \ldots X_{d-1}$ in $\vecx_0,\ldots, \vecx_{d-1}$ variables such that $f = \tr(X_0 \ldots X_{d-1})$. Hence, $h(\vecy) = \tr(X_0\cdot X_1\cdot X_2\cdot X_3(\veca_3)\ldots  X_{d-1}(\veca_{d-1}))$, where $X_k(\veca_k)$ is $X_k$ with its entries evaluated at $\veca_k$. Since $X_k$ is a full-rank linear matrix, with high probability, $X_k(\veca_k)\in \GL(w,\F)$. Then $M = X_3(\veca_3) \ldots X_{d-1}(\veca_{d-1}) \in \GL(w,\F)$, and $h(\vecy) = \tr(X_0\cdot X_1\cdot (X_2 M))$ is isomorphic to $\IMM_{w,3}$. At Step 2, $\MMTI$ returns $B_0,B_1,B_2 \in \GL(w^2,\F)$ such that $h(\vecy) = \IMM_{w,d}(B_0\vecx_0, B_1\vecx_1, B_2\vecx_2)$. It now follows from Lemma \ref{fact: symmetries of IMM} that corresponding to $X'_0, X'_1, X'_2$  there are matrices $C_0, C_1, C_2 \in \GL(w,\F)$ such that $X'_0 = C_0^{-1}\cdot X_0 \cdot C_1$, $X'_1 = C_1^{-1}\cdot X_1 \cdot C_2$, and $X'_2 = C_2^{-1}\cdot (X_2 M) \cdot C_0$. \vspace{0.1in}



\textbf{Steps 3--4}: The $d=4$ case arises in Algorithm \ref{algorithm: reduction from fmai to trace equivalence}. We present this case separately as it is easier to handle. Since $X'_0, X'_1, X'_2$ are full-rank linear matrices, at Step 3 a point $\vecb_{i,j}^{(k)}$ exists such that the $(i,j)$-th entry of $X_k(\vecb_{i,j}^{(k)})$ is one and other entries are zero. The point $\vecb_{i,j}^{(k)}$ can be computed by solving a system of linear equations. If $d=4$ then $f = \tr(X'_0\cdot X'_1\cdot X'_2\cdot (C_0^{-1} M^{-1})\cdot X_3\cdot C_0))$. Verify that $f(\vecb_{j,1}^{(0)}, \vecb_{1,1}^{(1)}, \vecb_{1,i}^{(2)}, \vecx_3)$ is the $(i,j)$-th entry of $(C_0^{-1}M^{-1})\cdot X_3\cdot C_0$. \vspace{0.1in} 

%
%
\textbf{Steps 5--8}: Let $Y = (C_0^{-1}M^{-1})\cdot X_3 \ldots X_{d-1} C_0$. Observe that $f(\vecx) = \tr(X'_0\cdot X'_1\cdot X'_2\cdot Y)$. At Step 5, $g(\vecx\setminus \vecy)$ is equals the $(1,1)$ entry of $Y$. Let $Y_3 = (C_0^{-1} M^{-1})\cdot X_3$, $Y_k = X_k$ for $k\in [4,d-2]$, and $Y_{d-1} = X_{d-1}\cdot C_0$, and $Y_3(i,*)$ and $Y_{d-1}(*,j)$ denote the $i$-th row of $Y_3$ and the $j$-th column of $Y_{d-1}$ respectively. Then $g(\vecx\setminus \vecy)$ is computed by the full-rank $(w,d,n)$-set-multilinear ABP $Y_3(1,*)\cdot Y_4\ldots Y_{d-2}\cdot Y_{d-1}(*,1)$ in $\vecx_3, \ldots, \vecx_{d-1}$ variables. Using the algorithm in \cite{KlivansS03}, a full-rank $(w,d,n)$-set-multilinear ABP $Y'_3 \cdot Y'_4 \cdots Y'_{d-2} \cdot Y'_{d-1}$ in $\vecx_3, \ldots ,\vecx_{d-1}$ variables is constructed that computes $g$. The ABP constructed is such that there are matrices $C_3, \ldots, C_{d-2} \in \GL(w,\F)$ such that $Y'_3 = Y_3(1,*)\cdot C_3$, $Y'_k = C_{k-1}^{-1}\cdot Y_k\cdot C_{k}$ for $k\in [4,d-2]$, and $Y'_{d-1} = C_{d-2}^{-1}\cdot Y_{d-1}(*,1)$. At Step 6, the points $\vecb_j^{(d-1)}$ and $\vecb_{i,j}^{(k)}$ are computed by solving systems of linear equations. Verify that $f = \textnormal{Trace}(X'_0\cdot X'_1\cdot X'_2 \cdot (Y_3 C_3)\cdot X'_4  \cdots X'_{d-2}\cdot (C_{d-2}^{-1} Y_{d-1}))$. This implies that $f(\vecb_{1,1}^{(0)}, \vecb_{1,1}^{(1)}, \vecb_{1,i}^{(0)}, \vecx_3, \mathbf{b}_{j,j}^{(4)}, \mathbf{b}_{j,j}^{(5)}, \ldots, \mathbf{b}_{j}^{(d-1)})$ is the $(i,j)$-th entry of $Y_3C_3$ at Step 7. Further, $f(\vecb_{j,1}^{(0)}, \vecb_{1,1}^{(1)}, \ldots, \vecb_{1,i}^{(d-2)}, \vecx_{d-1})$ is the $(i,j)$-th entry of $C_{d-2}^{-1} Y_{d-1}$ at Step 8. Hence, $X'_3 = Y_3C_3$ and $X'_{d-1} = C_{d-2}^{-1}Y_{d-1}$. In particular, $f = \tr(X'_0 \ldots X'_{d-1}) = \IMM_{w,d}(B_0\vecx_0, \ldots, B_{d-1}\vecx_{d-1})$. 

%% file: secappendix_reduction_from_FMAI_to_IMM.tex
\section{Proofs from Section \ref{section: reduction from fmai to imm}}\label{secappendix: reduction from fmai to imm}
\emph{Lemma \ref{lemma: characterization of trace by its lie algebra} (restated): Let $f$ be a non-zero $d$-tensor in the variable sets $\vecx_0, \ldots ,\vecx_{d-1}$ such that for all $k \in [0,d-1]$ $\CAL{B}_k \subseteq \G_f$. Then there is an $\alpha \in \F^{\times}$ such that $f(\vecx) = \alpha \cdot \IMM_{w,d}(\vecx)$.
}
\begin{proof}
A path monomial (see Definition \ref{definition :path monomial}) looks like $\mu = x_{i_0,i_1}^{(0)}\cdot x_{i_1,i_2}^{(1)} \ldots x_{i_{d-1},i_0}^{(d-1)}$. The next claim shows that the coefficient of a non-path monomial in $f$ is zero. In the claim, $f_{(u,v),(p,q)}^{(k,k+1)}$ denotes the coefficient of $x_{u,v}^{(k)}x_{p,q}^{(k+1)}$ in $f$ over $\F[\vecx\setminus \{\vecx_k,\vecx_{k+1}\}]$, where $u,v,p,q \in [w]$.
\begin{claim}\label{claim: non-path monomials have zero coefficient}
Let $\mu$ be a non-path monomial. Then the coefficient of $\mu$ in $f$ is zero.
\end{claim}
\begin{proof}
Let $\mu = x_{i_0,j_0}^{(0)}\cdot x_{i_1,j_1}^{(1)}\ldots  x_{i_{d-1},j_{d-1}}^{(d-1)}$ be a non-path monomial. Hence there is a $k \in [0,d-1]$ such that $j_{k} \neq i_{k+1}$. Suppose $k \in [0,d-2]$, and $k$ is even.  Let $D \in \CAL{M}_w$ be a diagonal matrix such that its $(j_k,j_k)$ entry is one and all the other entries are zero. Let $B \in \CAL{B}_k$ be a block-diagonal matrix  whose $2w^2 \times 2w^2$ sub-matrix indexed by $\vecx_k\uplus \vecx_{k+1}$ looks like 
\begin{equation*}
    \begin{bmatrix}
     I_w \otimes D & \mathbf{0} \\
     \mathbf{0}  & -I_w \otimes D
    \end{bmatrix} ~~ .
\end{equation*}
Since $B \in \G_f$, we have (by the variable ordering in $\vecx_k$ and $\vecx_{k+1}$),
\begin{equation}\label{equation: coefficient equation for non path monomials}
    \sum_{u \in [w]} x_{u,j_k}^{(k)} \left(\sum_{p,q \in [w]} x_{p,q}^{(k+1)} f_{(u,j_k),(p,q)}^{(k,k+1)}\right) = \sum_{u \in [w]} x_{j_k,u}^{(k+1)} \left(\sum_{p,q \in [w]} x_{p,q}^{(k)} f_{(p,q),(j_k,u)}^{(k,k+1)}\right).
\end{equation}
From Equation \ref{equation: coefficient equation for non path monomials} we conclude that for all $u,q \in [w]$, $f_{(u,j_k),(p,q)}^{(k,k+1)} =0$ if $p \neq j_k$. Now suppose for contradiction the coefficient of $\mu$ in $f$ is non-zero. Then the coefficient of $x_{i_k,j_k}^{(k)}x_{i_{k+1},j_{k+1}}^{(k+1)}$ (i.e., $f_{(i_k,j_k),(i_{k+1},j_{k+1})}^{(k,k+1)}$) is non-zero. Since $j_k \neq i_{k+1}$, this is a contradiction. If $k \in [0,d-1]$ and $k$ is odd then the proof is similar and the only thing to note in this case is that $B \in \CAL{B}_k$ is such that its $2w^2 \times 2w^2$ sub-matrix indexed by the $\vecx_k\uplus \vecx_{k+1}$ looks like
\begin{equation*}
    \begin{bmatrix}
     D \otimes I_w & \mathbf{0} \\
     \mathbf{0}  & -D \otimes I_w
    \end{bmatrix} ~~ .
\end{equation*}
Finally, if $k = d-1$ and $d$ is odd then again the proof is similar but $B \in \CAL{B}_{d-1}$ in this case is such that the $w^2 \times w^2$ sub-matrix of $B$ indexed by $\vecx_{d-1}$ variables is equal to $I_w\otimes D$, and the $w^2 \times w^2$ sub-matrix of $B$ indexed by $\vecx_0$ variables is equal to $-D\otimes I_w$.   
\end{proof}
\begin{claim}\label{claim: path monomials from adjacent sets have the same coefficient}
Let $\mu_1 = x_{i_0,i_1}^{(0)}\cdot x_{i_1,i_2}^{(1)}\ldots x_{i_k,i_{k+1}}^{(k)}\cdot x_{i_{k+1},i_{k+2}}^{(k+1)}\ldots x_{i_{d-1},i_0}^{(d-1)}$, and $\mu_2 = x_{i_0,i_1}^{(0)}\cdot x_{i_1,i_2}^{(1)}\ldots x_{i_k,i'_{k+1}}^{(k)}\cdot x_{i'_{k+1},i_{k+2}}^{(k+1)}\ldots x_{i_{d-1},i_0}^{(d-1)}$ be two path monomials. Then the coefficients of $\mu_1$ and $\mu_2$ in $f$ are equal. 
\end{claim}
\begin{proof}
Suppose $k \in [0,d-2]$ and $k$ is even. Let $M \in \CAL{M}_w$ be such that its $(i'_{k+1},i_{k+1})$ is one and all its other entries are zero, and $B \in \CAL{B}_k$ be a block-diagonal matrix such that $B$ restricted to the $2w^2 \times 2w^2$ sub-matrix indexed by $\vecx_{k}\uplus \vecx_{k+1}$ variables is as shown below
\begin{equation*}
    \begin{bmatrix}
     I_w \otimes M^T & \mathbf{0} \\
     \mathbf{0}  & -I_w \otimes M
    \end{bmatrix} ~~ .
\end{equation*}
Let the coefficients of $\mu_1$ and $\mu_2$ in $f$ be equal to $\alpha_1$ and $\alpha_2$. Since $B \in \G_f$, we have 
\begin{equation}\label{equation: coefficient equation for adjacent path monomials}
    \sum_{u \in [w]} x_{u,i'_{k+1}}^{(k)} \frac{\partial f}{\partial x_{u,i_{k+1}}^{(k)}}~ - ~\sum_{u \in [w]} x_{i_{k+1},u}^{(k+1)} \frac{\partial f}{\partial x_{i'_{k+1},u}^{(k+1)}} = 0.
\end{equation}
The coefficient of $\mu = x_{i_0,i_1}^{(0)}\cdot x_{i_1,i_2}^{(1)}\ldots x_{i_k,i'_{k+1}}^{(k)}\cdot x_{i_{k+1},i_{k+2}}^{(k+1)}\ldots x_{i_{d-1},i_0}^{(d-1)}$ in Equation \ref{equation: coefficient equation for adjacent path monomials} is equal to $\alpha_1 - \alpha_2 = 0$. Hence $\alpha_1 = \alpha_2$. The proof for the two remaining cases: a) $k \in [0,d-1]$ and $k$ odd, and b) $k = d-1$ and $d$ odd, follow similarly by constructing appropriate $B$ matrices. 
\end{proof}
We now use above the claim to show the following. 
\begin{claim}\label{claim: path monomials have the same coefficient}
Let $\mu_1 = x_{i_0,i_1}^{(0)}\cdot x_{i_1,i_2}^{(1)}\ldots x_{i_{d-1},i_0}^{(d-1)}$, and $\mu_2 = x_{j_0,j_1}^{(0)}\cdot x_{j_1,j_2}^{(1)}\ldots x_{j_{d-1},j_0}^{(d-1)}$ be two path monomials. Then the coefficient of $\mu_1$ and $\mu_2$ in $f$ are equal. 
\end{claim}
\begin{proof}
For $k \in [1,d-2]$, let $\nu_k = x_{i_0,j_1}^{(0)}\cdot x_{j_1,j_2}^{(1)} \ldots x_{j_{k-1},j_k}^{(k-1)} \cdot x_{j_k,i_{k+1}}^{(k)}\cdot x_{i_{k+1},i_{k+2}}^{(k+1)}\ldots x_{i_{d-1},i_{0}}^{(d-1)}$, and  $\nu_{d-1} = x_{i_0,j_1}^{(0)}\cdot x_{j_1,j_2}^{(1)} \ldots x_{j_{d-2},j_{d-1}}^{(d-2)} \cdot x_{j_{d-1},i_0}^{(d-1)}$. From Claim \ref{claim: path monomials from adjacent sets have the same coefficient}, the coefficients of $\mu_1$ and $\nu_1$ in $f$ are equal, the coefficients of $\nu_k$ and $\nu_{k+1}$ in $f$ are equal for $k \in [1,d-2]$, and the coefficients of $\nu_{d-1}$ and $\mu_2$ in $f$ are equal. Hence the coefficients of $\mu_1$ and $\mu_2$ in $f$ are equal.  
\end{proof}
It follows immediately that there is an $\alpha \in \F^{\times}$ such that $f = \alpha \cdot \IMM_{w,d}(\vecx)$. 
\end{proof}

\emph{Corollary \ref{corollary: lie algebra characterization for equivalent polynomials} (restated): Let $B \in \mathsf{GL}(n,\F)$ be a block-diagonal matrix with individual blocks $B_0, \ldots, B_{d-1}$ and  $f$ be a non-zero $d$-tensor in the variable sets $\vecx_0, \ldots ,\vecx_{d-1}$ such that for all $k \in [0,d-1]$, $B^{-1}\cdot \CAL{B}_k \cdot B \subseteq \G_f$. Then there is an $\alpha \in \F^{\times}$ such that $f(\vecx) = \alpha \cdot \IMM_{w,d}(B_0\vecx_0,\ldots,$  $B_{d-1}\vecx_{d-1})$.}
\begin{proof}
Let $g(\vecx) = f(B_0^{-1}\vecx_0, \ldots , B_{d-1}^{-1}\vecx_{d-1})$. Since $B^{-1}\cdot \CAL{B}_k\cdot B \subseteq \G_f$, $\CAL{B}_{k} \subseteq \G_g$ for $k\in [0,d-1]$ (from Fact \ref{fact: lie algebra conjugacy}). Hence, from Lemma \ref{lemma: characterization of trace by its lie algebra} there is an $\alpha \in \F^{\times}$ such that $g = \alpha\cdot \IMM_{w,d}$. Thus $f(\vecx) = \alpha\cdot \IMM_{w,d}(B_0\vecx_0, \ldots ,  $ $B_{d-1}\vecx_{d-1})$. 
\end{proof}

\emph{Claim \ref{claim: left multiplication algebra of the given algebra} (restated): Suppose $\CAL{A} \cong \CAL{M}_w$ for some $w \in \N$. Then there exists a $K \in \GL(w^2, \F)$ and linearly independent matrices $\{C_{1,1}, \ldots , C_{w,w}\}$ in $\CAL{M}_w$  such that $L_{i,j} = K^{-1}\cdot (I_w \otimes C_{i,j}) \cdot K$~ for all $i,j \in [w]$.}
\begin{proof}
Let $\CAL{L}$ be the algebra generated by the matrices $\{L_{1,1}, \ldots , L_{w,w}\}$. It is easy to see that $\CAL{A} \cong \CAL{L} \cong \CAL{M}_w$ and $\CAL{L}$ contains $I_{w^2}$. 
From Skolem-Noether theorem (see Theorem 5 in \cite{GargGKS19}, and \cite{Lorenz08}) we have that there is a $K \in \mathsf{GL}(w^2, \F)$ and linearly independent matrices $\{C_{1,1}, \ldots , C_{w,w}\}$ in $\CAL{M}_w$  such that $L_{i,j} = K^{-1}\cdot (I_w \otimes C_{i,j}) \cdot K$ for all $i,j \in [w]$.
\end{proof}

